\documentclass[twoside,leqno,twocolumn]{article}  
\usepackage[varg]{txfonts}
\usepackage{ltexpprt}

\usepackage{amsmath}
\usepackage{algorithmicx}
\usepackage[ruled,vlined]{algorithm2e}
\usepackage{booktabs}
\usepackage{graphicx}
\usepackage{xspace}
\usepackage{color}
\usepackage[colorlinks]{hyperref}
\usepackage{url}
\usepackage{tikz}
\usepackage{enumitem}
\usetikzlibrary{positioning}

\usepackage{amssymb,amsfonts,amsmath}

\usepackage[mathscr]{eucal}

\usepackage{stmaryrd}

\usepackage{amsbsy}

\usepackage{bm}

\usepackage{braket}

\newcommand{\Tra}{{\sf T}}

\definecolor{green}{RGB}{20, 150, 60}

\pgfdeclarelayer{bg}    
\pgfsetlayers{bg,main}  

\begin{document}

\newcommand{\TODO}[1]{\textcolor{blue}{\textbf{TODO XXXXX: #1}}}
\newcommand{\eqnref}[1]{Eqn.~\eqref{#1}}
\newcommand{\figref}[1]{Figure~\ref{#1}}
\renewcommand{\cond}[1]{\phi\left( #1 \right)}
\newcommand{\condthree}[1]{\phi_3\left( #1 \right)}
\newcommand{\diag}[1]{\text{diag}\left( #1 \right)}
\newcommand{\cut}[1]{\text{cut}\left(#1\right)}
\newcommand{\vol}[1]{\text{vol}(#1)}
\newcommand{\volthree}[1]{\text{vol}_3(#1)}
\newcommand{\cutthree}[1]{\text{cut}_3(#1)}
\newcommand{\mat}[1]{\bm{#1}}
\newcommand{\tens}[1]{\bm{#1}}
\renewcommand{\tens}[1]{\underline{\bm{#1}}}
\newcommand{\vect}[1]{\bm{#1}}
\newcommand{\allones}{\vect{e}}
\newcommand{\Flow}[1]{F\left(#1\right)}
\newcommand{\ncut}[1]{\text{ncut}\left(#1\right)}
\newcommand{\sign}[1]{\text{sign}(#1)}
\newcommand{\tenstrans}{\mat{P}[\vect{x}]}
\newcommand{\dataset}[1]{\texttt{#1}}
\newcommand{\hide}[1]{}
\newcommand{\tsc}{{\sc{TSC}}\xspace}
\newcommand{\eg}{\emph{e.g.}}
\newcommand{\ie}{\emph{i.e.}}

\graphicspath{{./FIG/}}

\title{\Large Tensor Spectral Clustering for Partitioning Higher-order Network Structures}
\author{
Austin R. Benson\thanks{
Institute for Computational and Mathematical Engineering, Stanford University.}
\and 
David F. Gleich\thanks{Department of Computer Science, Purdue University.}
\and
Jure Leskovec\thanks{Department of Computer Science, Stanford University.}
}
\date{}

\maketitle


\begin{abstract} 

Spectral graph theory-based methods represent an important class of tools for studying the structure of networks. 
Spectral methods are based on a first-order Markov chain derived from a random walk on the graph and thus they cannot take advantage of important higher-order network substructures such as triangles, cycles, and feed-forward loops.
Here we propose a {\em Tensor Spectral Clustering} (\tsc) algorithm that allows for modeling higher-order network structures in a graph partitioning framework. 
Our \tsc algorithm allows the user to specify which higher-order network structures (cycles, feed-forward loops, etc.)~should be preserved by the network clustering.
Higher-order network structures of interest are represented using a tensor, which we then partition 
by developing a multilinear spectral method.
Our framework can be applied to discovering layered flows in networks as well as graph anomaly detection,
which we illustrate on synthetic networks.
In directed networks, a higher-order structure of particular interest is the directed 3-cycle, which captures feedback loops in networks.
We demonstrate that our \tsc algorithm produces large partitions that cut fewer directed 3-cycles than standard spectral clustering algorithms.

\hide{
Traditional spectral graph partitioning uses a first-order Markov chain derived from random walks on the graph.
Because the Markov chain is first-order, these spectral algorithms cannot take advantage of important higher-order network structures such as triangles and cycles.
We propose a tensor spectral partitioning algorithm that uses tensors of higher-order network data and approximations to higher-order Markov chains to directly partition higher-order structures.
In directed networks, one higher-order structure of particular interest is the directed 3-cycle, which arises in directed community detection and as a motif for network feedback.
We empirically demonstrate that our multilinear clustering algorithm produces partitions that do not cut many directed 3-cycles compared to standard spectral clustering algorithms.
We compliment our empirical findings with theoretical justification and illustrative examples on small networks.
}

\end{abstract}


\section{Introduction}
\label{sec:010intro}

Spectral graph methods investigate the structure of networks by studying the eigenvalues and eigenvectors of matrices associated to the graph, such as its adjacency matrix or Laplacian matrix. Arguably the most important spectral graph algorithms are the spectral graph partitioning methods that identify partitions of nodes into low conductance communities in undirected networks~\cite{alon1985lambda}. 
While the simple matrix computations and strong mathematical theory behind spectral clustering methods
makes them appealing, the methods are inherently limited to {\em two-dimensional} structures, for example, undirected edges connecting {\em pairs} nodes. 
Thus, it is a natural question whether spectral methods can be generalized to higher-order network structures. For example, traditional spectral clustering attempts to minimize (appropriately normalized) number of first-order structures (\ie, edges) that need to be cut in order to split the graph into two parts. In a similar spirit, a higher-order generalization of spectral clustering would try to minimize cutting {\em higher-order structures} that involve multiple nodes (\eg, triangles).


Incorporating higher-order graph information (that is, network motifs/graphlets) into the partitioning process can significantly improve our understanding of the underlying network.
For example, triangles (three-dimensional network structures involving \emph{three} nodes) have proven fundamental to understanding social networks~\cite{granovetter1973strength,kossinets2006empirical} and their community structure~\cite{durak2012degree,prat2012shaping,rosvall2014memory}.
Most importantly, higher-order spectral clustering would allow for greater modeling flexibility as the application would drive which higher-order network structures should be preserved by the network clustering. For example, in financial networks, directed cycles might indicate money laundering and
higher-order spectral clustering could be used to identify groups of nodes that participate in such directed cycles. As directed cycles involve multiple edges, current spectral clustering tools would not be able to identify groups with such structural signatures.


Generalizing spectral clustering to higher-order structures involves 
several challenges. The essential challenge is that higher-order 
structures are often encoded in tensors, 
i.e., multi-dimensional matrices. Even simple computations with tensors
lack the traditional algorithmic guarantees of two-dimensional matrix
computations such as existence and known runtimes. For instance, 
eigenvectors are a key component to spectral clustering, and finding 
tensor eigenvectors is NP-hard~\cite{hillar2013most}.
An additional challenge is 
that the number of higher-order structures increases exponentially 
with the size of the structure. For example, in a graph with n nodes, 
the number of possible triangles is $O(n^3)$.
However, real-world networks have far fewer triangles.


While there exist several extensions to the spectral method, including the directed Laplacian~\cite{chung2005laplacians}, the asymmetric Laplacian~\cite{boley2011commute}, and co-clustering~\cite{dhillon2001co,rohe2012co}, these methods are all limited to two-dimensional graph representations. A simple work-around would be to weight edges that occur in higher-order structures~\cite{klymko2014using}. However, this heuristic is unsatisfactory because the optimization is still on edges,
and not on the higher-order patterns we aim to cluster.


Here, we propose a {\em Tensor Spectral Clustering (\tsc)} framework that is directly based on higher-order network structures, \ie, network information beyond edges connecting two nodes.
Our framework operates on a tensor of network data and allows the user to specify which higher-order network structures (cycles, feed-forward loops, etc.)~should be preserved by the clustering. For example, if  one aims to obtain a partitioning that does not cut triangles, then this can be encoded in a third-order tensor $\tens{T}$, where $\tens{T}(i, j, k)$ is equal to $1$ if nodes $i$, $j$, and $k$ form a triangle and $0$ otherwise. 

Given a tensor representation of the desired higher-order network structures, we then use a mutlilinear {PageRank} vector~\cite{gleich2014multilinear} to reduce the tensor to a two-dimensional matrix. This dimensionality reduction step allows us to use efficient matrix algorithms while approximately preserving the higher-order structures represented by the tensor. Our resulting \tsc algorithm is a spectral method that partitions the network to minimize the number of higher-order structures cut. This way our algorithm finds subgraphs that contain many instances of the higher-order structure described by the tensor. 
Figure~\ref{fig:simple_example} illustrates a directed network, and our goal is to identify clusters of directed 3-cycles. That is, we aim to partition the nodes into two sets such that few directed 3-cycles get cut. Our \tsc algorithm finds a partition that does not cut any of the directed 3-cycles, while a standard spectral partitioner (the directed Laplacian~\cite{chung2005laplacians})~does.


Clustering networks based on higher-order structures has many applications. For example, the \tsc algorithm allows for identifying layered flows in networks, where the network consists of several layers that contain many feedback loops. Between layers, there are many edges, but they flow in one direction and do not contribute to feedback. We identify such layers by clustering a tensor that describes small feedback loops (\emph{e.g.}, directed 3-cycles and reciprocated edges). Similarly, \tsc can be applied to anomaly detection in directed networks, where the tensor encodes directed 3-cycles that have no reciprocated edges. Our \tsc algorithm can find subgraphs that have many instances of this pattern, while other spectral methods fail to capture these higher-order network structures.

Our contributions are summarized as follows:
\begin{itemize}[noitemsep,nolistsep]
\item
In Sec.~\ref{sec:030framework}, we develop a tensor spectral clustering framework that computes directly on higher-order graph structures.
We provide theoretical justifications for our framework in Sec.~\ref{sec:040generalizations}.
\item
In Sec.~\ref{sec:050applications}, we provide two applications---layered flow networks and anomaly detection---where
our tensor spectral clustering algorithm outperforms standard spectral clustering on small, illustrative networks.
\item
In Sec.~\ref{sec:060d3c_large}, we use tensor spectral clustering to partition large networks so that directed 3-cycles are not cut.
This provides additional empirical evidence that our algorithm out-performs state-of-the-art spectral methods.
\end{itemize}

Code used for this paper is available at \url{https://github.com/arbenson/tensor-sc}, and all networks used in experiments are available from SNAP~\cite{snapnets}.

\begin{figure}
\begin{minipage}{\textwidth}
\begin{minipage}{0.24\textwidth}
\tikzset{first node/.style={circle,fill=blue!20,draw,minimum size=0.6cm,inner sep=0pt},}
\tikzset{second node/.style={circle,fill=red!20,draw,minimum size=0.6cm,inner sep=0pt},}
\scalebox{0.7}{
 \begin{tikzpicture}
    \node[first node] (0) {$0$};
    \node[first node] (1) [below right = 0.25cm and 3cm of 0]  {$1$};
    \node[first node] (2) [below right = 0.75cm and 0.5cm of 0] {$2$};
    \node[second node] (3) [below = 0.5cm of 2] {$3$};
    \node[second node] (4) [below right = 0.25cm and 2cm of 3]  {$4$};
    \node[second node] (5) [below left = 0.75cm and 0.1cm of 3] {$5$};    

    \path[draw, ultra thick] (0) edge [->] (1);
    \path[draw, ultra thick] (1) edge [->] (2);
    \path[draw, ultra thick] (2) edge [->] (0);
    \path[draw, ultra thick] (1) edge [->] (0);
    \path[draw, ultra thick] (2) edge [->] (1);
    \path[draw, ultra thick] (0) edge [->] (2);    
    
    \path[draw, ultra thick] (3) edge [->] (4);
    \path[draw, ultra thick] (4) edge [->] (5);
    \path[draw, ultra thick] (5) edge [->] (3);    
    
    \path[draw, ultra thick] (2) edge [->] (3);
    \path[draw, ultra thick] (2) edge [->] (4);
    \path[draw, ultra thick] (1) edge [->] (4);    
    
    \path[draw, ultra thick] (5) edge [->] (0);        

\end{tikzpicture}
}
\end{minipage}
\begin{minipage}{0.3\textwidth}
\begin{tabular}{l}
Tensor spectral clustering: \\
\{\textcolor{blue}{0}, \textcolor{blue}{1},  \textcolor{blue}{2}\},
\{\textcolor{red}{3}, \textcolor{red}{4},  \textcolor{red}{5}\} \\
\\
Directed Laplacian: \\
\{\textcolor{blue}{1}, \textcolor{blue}{2},  \textcolor{red}{5}\},
\{\textcolor{blue}{0}, \textcolor{red}{3},  \textcolor{red}{4}\}
\end{tabular}
\end{minipage}
\end{minipage}
\caption{
(Left)
Network where directed 3-cycles only appear within the blue or red nodes.
(Right)
Partitioning found by our proposed tensor spectral clustering algorithm
and the directed Laplacian.
Our proposed algorithm doesn't cut any directed 3-cycles.
Directed 3-cycles are just one higher-order structure that can be used within our framework.
}
\label{fig:simple_example}
\end{figure}

\section{Preliminaries and background}
\label{sec:020preliminaries}
We now review spectral clustering and conductance cut.
The key ideas are a Markov chain representing a random walk on a graphs,
a second left eigenvector of the Markov chain,
and a sweep cut that uses the ordering of the eigenvector to compute conductance scores.
In Sec.~\ref{sec:030framework}, we generalize these ideas to tensors and higher-order structures on graphs.

\subsection{Notation and the transition matrix}
\label{sec:021transition_matrix}
Consider an undirected, weighted graph $G = (V, E)$, where $n = |V|$ and $m = |E|$.
Let $\mat{A} \in \mathbb{R}^{n \times n}_+$ be the weighted adjacency matrix of $G$, \emph{i.e.}, $A_{ij} = w_{ij}$ if $(i, j) \in E$ and $A_{ij} = 0$ otherwise.
Let $\mat{D}$ be the diagonal matrix with generalized degrees of the vertices of $G$.
In other words, $\mat{D} = \diag{\mat{A}\allones}$, where $\allones$ is the vector of all ones.
The \emph{combinatorial Laplacian} or \emph{Kirchoff} matrix is $\mat{K} = \mat{D} - \mat{A}$.
The matrix $\mat{P} = \mat{A}^{\Tra}\mat{D}^{-1}$ is a column stochastic matrix,
which we call the \emph{transition matrix}.
We now interpret this matrix as a Markov chain.

\subsection{Markov chain interpretation}
\label{sec:022markov}
Since $\mat{P}$ is column stochastic,
we can interpret the matrix as a Markov chain with states $S_t$, for each time step $t$.
Specifically, the states of the Markov chain are the vertices on the graph, \ie, $S_t \in V$.
The transition probabilities are given by $\mat{P}$: 
\[
\text{Prob}(S_{t+1} = i \;\mid\; S_{t} = j) = \mat{P}_{ij} = A_{ji} / D_{jj}.
\]

This Markov chain represents a \emph{random walk} on the graph $G$.
In Sec.~\ref{sec:032markov}, we will generalize this idea to tensors of graph data.
We now show how the second left eigenvector of the Markov chain described here is key to spectral clustering.

\subsection{Second left eigenvector for conductance cut}
\label{sec:023evec_conductance}

The \emph{conductance} of a set $S \subset V$ of nodes is
\begin{equation}\label{eqn:conductance}
\cond{S} = \cut{S} / \min\left(\vol{S}, \vol{\bar{S}} \right),
\end{equation}
where $\cut{S} = \left\vert \{ (u, v) \;\mid\; u \in S, v \in \bar{S} \} \right\vert$,
and $\vol{S} = \left\vert \{ (u, v) \;\mid\; u \in S \} \right\vert$.
Small conductance indicates a good partition of the graph: the number of cut edges must be small and neither $S$ nor $\bar{S}$ can be too small.
Let $\vect{z} \in \{ -1, 1 \}^n$ be an indicator vector over the nodes in $G$, where $z_i = 1$ if the $i$th node is in $S$.
Then
\begin{equation}\label{eqn:cut}
\vect{z}^{\Tra}\mat{K}\vect{z} = \sum_{(i, j) \in E} 4\mathbb{I}\left(z_i = z_j\right) \propto \cut{S}.
\end{equation}

The conductance cut eigenvalue problem is an approximation for the NP-hard problem of minimizing conductance:
\begin{equation}\label{eqn:conductance_cut}
\begin{aligned}
& \underset{\vect{z} \in \mathbb{R}^n}{\text{minimize}}
& & \vect{z}^{\Tra}\mat{K}\vect{z} / \vect{z}^{\Tra}\mat{D}\vect{z} \\
& \text{subject to}
& & \allones^{\Tra}\mat{D}\vect{z} = 0, \quad \| \vect{z} \| = 1
\end{aligned}
\end{equation}

The idea of the real-valued relaxation in \eqnref{eqn:conductance_cut} is that positive and negative values of $\vect{z}$ correspond to the $\pm 1$ indicator vector for the cut in \eqnref{eqn:cut}.
In Sec.~\ref{sec:024sweep} we will review how to convert the real-valued solution to a cut.

The matrices $\mat{K}$ and $\mat{D}$ are positive semi-definite, and \eqnref{eqn:conductance_cut} is a generalized eigenvalue problem.
In particular, the solution is the vector $\vect{z}$ such that $\mat{K}\vect{z} = \lambda \mat{D}\vect{z}$, where $\lambda$ is the second smallest generalized eigenvalue (the smallest eigenvalue is $0$ and corresponds to the trivial solution $\vect{z} = \allones$).
To get the solution $\vect{z}$, we observe that
\begin{eqnarray*}
\mat{K}\vect{z} = \lambda \mat{D}\vect{z}
&\iff&  (\mat{I} - \mat{D}^{-1}\mat{A})\vect{z} = \lambda \vect{z} \\
&\iff& \vect{z}^T\mat{P} = (1 - \lambda)\vect{z}^T
\end{eqnarray*}
where $1 - \lambda$ is the second largest left eigenvalue of $\mat{P}$.
We know that $\allones^T\mat{P} = \allones^T$, so we are looking for the dominant left eigenvector that is orthogonal to the trivial one.

Here, we call the above partitioning algorithm for undirected graphs the ``undirected Laplacian'' method.
One generalization to directed graphs is due to Chung~\cite{chung2005laplacians}.
For this method, we use the undirected Laplacian method on the following symmetrized network:
$\mat{A}_{sym} := \frac{1}{2}\left(\mat{\Pi} \mat{P}^{\Tra} + \mat{P} \mat{\Pi}\right)$,
where $\mat{P} = \mat{A}^{\Tra}\mat{D}^{-1}$ and $\mat{\Pi} = \diag{\pi}$ for $\mat{P}\pi = \pi$, the stationary distribution of $\mat{P}$.
Note that $\mat{D}_{sym} = \diag{\mat{A}_{sym}\allones} = \mat{\Pi}$,
so we are interested in the second left eigenvector of
\begin{equation}\label{eqn:P_sym}
\mat{P}_{sym} = \frac{1}{2}\left(\mat{\Pi}\mat{P}^{\Tra}\mat{\Pi}^{-1} + \mat{P}\right).
\end{equation}
By ``directed Laplacian", we refer to the method that uses the second left eigenvector of $\mat{P}_{sym}$.

\subsection{Sweep cuts}
\label{sec:024sweep}
In order to round the real-valued solution $\vect{z}$ to a solution set $S$ to evaluate \eqnref{eqn:conductance},
we sort the vertices by the values $z_i$ and consider vertex sets $S_k$ that consist of the first $k$ nodes in the sorted vertex list.
In other words, if $\sigma_i$ is equal to the index of the $i$th smallest element of $\vect{z}$,
then $S_k = \{\sigma_1, \sigma_2, \ldots \sigma_k \}$.
We then choose $S = \arg\min_{S_k} \phi(S_k)$.
The set of nodes $S$ satisfies the celebrated Cheeger inequality~\cite{alon1985lambda}: $\phi_*^2 / 2 \le \phi(S) \le 2\phi_*$, where $\phi_*$ is the minimum conductance over all cuts.
The sweep cut computation is fast, since $S_{k+1}$ differs from $S_k$ by only one node,
and the sequence of scores $\phi(S_{1}), \ldots, \phi(S_{n})$ can be computed in $O(n + m)$ time.

In addition to conductance, other scores can also be computed in the same sweeping fashion.
Of particular interest are the normalized cut,
$ncut(S) = \cut{S}\left(1 / \vol{S} + 1 / \vol{\bar{S}}\right)$, 
and the expansion, $\rho(S) = \cut{S} / \min(|S|, |\bar{S}|)$.
The normalized cut differs by at most a factor of two from conductance,
so we will limit ourselves to conductance and expansion in this paper.

\section{Tensor spectral clustering framework}
\label{sec:030framework}

The key ingredients for spectral clustering discussed in Sec.~\ref{sec:020preliminaries} were
a transition matrix from an undirected graph,
a Markov chain interpretation of the transition matrix,
and the second left eigenvector of the Markov chain.
We now generalize these ideas for higher-order network structures.

\subsection{Transition tensors}
\label{sec:031transition_tensors}

Our first goal is to represent the higher-order network stuctures of interest. For example, to represent structures on three nodes (\ie, directed cycles, or feed-forward loops) we required a three-dimensional tensor. 
In particular, we want a \emph{symmetric} order-$3$ tensor $\tens{T} \in \mathbb{R}^{n \times n \times n}_+$
such that the entry at index $(i, j, k)$ contains information about nodes $i$, $j$, $k \in V$.
(Here, symmetric means that the value of $\tens{T}(i, j, k)$ remains the same under any permutation of the three indices.)
A tensor describing triangles in $G$ is:
\begin{equation}\label{eqn:tri_tens}
\tens{T}(i, j, k) = \mathbb{I}\left(i, j, k \in V \text{ distinct and form a triangle}\right).
\end{equation}

This tensor represents third-order information about the graph.
We form a transition tensor by
\[
\tens{P}(i, j, k) = \tens{T}(i, j, k) / \sum_{i=1}^{n}\tens{T}(i, j, k), \quad 1 \le i, j, k \le n.
\]
In the case that  $\sum_{i=1}^{n}\tens{T}(i, j, k) = 0$, we fill in $\tens{P}(:, j, k)$ with a stochastic vector $\vect{u}$, \ie, $\tens{P}(:, j, k) = \vect{u}$.
We call the vector $\vect{u}$ the \emph{dangling distribution vector}, borrowing the term from the PageRank community \cite{boldi2008traps}.
Next, we see how to interpret this transition tensor as a second-order Markov chain.

\subsection{Second-order Markov chains and the spacey random surfer}
\label{sec:032markov}

Next, we seek to generalize the Markov chain interpretation of spectral clustering to tensors. While spectral clustering on matrices is analogous to a first-order Markov chain, we will show that tensor spectral clustering is analogous to a second-order Markov chain on a matrix representation of the tensor.

Entries of the transition tensor $\tens{P}$ from Sec.~\ref{sec:031transition_tensors} can be interpreted as
the transition probabilities of a second-order Markov chain. Specifically, given a second-order Markov chain with state space the set of vertices, $V$, we define the transition probabilities as
\[
\tens{P}(i, j, k) = \text{Prob}\left(S_{t+1} = i \;\mid\; S_{t} = j, S_{t-1} = k\right).
\]
In other words, the probability of moving to state $i$ depends on the current state $j$ and the last state $k$. 
For the triangle tensor in \eqnref{eqn:tri_tens},
\[
\tens{P}(i, j, k) = \frac{\mathbb{I}\left(i, j, k \text{ form triangle}\right)}{\#(\text{triangles involving nodes $j$ and $k$})}
\]
If the previous state was node $k$ and the current state is node $j$, then, for the next state, the Markov chain chooses uniformly over all nodes $i$ that form a triangle with $j$ and $k$.

The stationary distribution $X_{ij}$ of the second-order Markov chain satisfies
$\sum_{k}\tens{P}(i, j, k)X_{jk} = X_{ij}$.
We would like to model the full second-order dynamics of the Markov chain, but doing so is computationally infeasible because just storing the stationary distribution requires $O(n^2)$ memory.
Instead, we will make the simplifying assumption that $X_{ij} = x_{i}x_{j}$ for some vector $\vect{x} \in \mathbb{R}_+^n$ with $\sum_{i} x_i = 1$.
The stationary distribution then satisfies
\begin{equation}\label{eqn:z_eigenpair}
\sum_{1 \le j, k \le n}\tens{P}(i, j, k)x_jx_k = x_i.
\end{equation}

With respect to \eqnref{eqn:z_eigenpair}, $\vect{x}$ is called a $z$ eigenvector of the tensor $\tens{P}$ with eigenvalue $1$~\cite{qi2005eigenvalues}.
To simplify notation, we will denote the one-mode unfolding of $\tens{P}$ by $\mat{R} \in \mathbb{R}^{n \times n^2}$, namely
$\mat{R} = \begin{bmatrix} \tens{P}(:, :, 1) & \tens{P}(:, :, 2) & \ldots & \tens{P}(:, :, n) \end{bmatrix}$.
The matrix $\mat{R}$ is a column stochastic matrix.
We use $\mat{R}_k = \tens{P}(:, :, k)$ to denote the $k$th $n \times n$ block of $\mat{R}$.
With this notation, \eqnref{eqn:z_eigenpair} reduces to $\mat{R} \cdot \left(\vect{x} \otimes \vect{x}\right) = \vect{x}$, where $\otimes$ denotes the Kronecker product.

The simplifying approximation $X_{ij} = x_ix_j$ is computationally and algebraically appealing, but we also want a random process to interpret the vector. Recent work~\cite{gleich2014multilinear} has considered the \emph{multilinear pagerank} vector $\vect{x}$ that satisfies
\begin{equation}\label{eqn:srs}
\alpha\mat{R}\left( \vect{x} \otimes \vect{x} \right) + (1 - \alpha)\vect{v} = \vect{x}, \; x_k \ge 0, \; \allones^{\Tra}\vect{x} = 1,
\end{equation}
for a constant $\alpha \in (0, 1)$ and stochastic vector $\vect{v}$.

This vector is the stationary distribution of a stochastic process recently termed the \emph{spacey random surfer}~\cite{gleich2014spacey}. At any step of the process, a random surfer has just moved from node $k$ to node $j$. With probability $(1 - \alpha)$, the surfer teleports to a random state via the stochastic vector $\vect{v}$. With probability $\alpha$, the surfer wants to transition to node $i$ with probability $\tens{P}(i, j, k)$. However, the surfer \emph{spaces out} and forgets that s/he came from node $k$.
Instead, the surfer guesses the previous state, based on the historical distribution over the state space.
Formally, the surfer guesses node $\ell$ with probability $\frac{1}{t + n}\left(1 + \sum_{r=1}^{t} \mathbb{I}\left[S_{t} = \ell\right]\right)$. It is important to note that although this process is an approximation to a second-order Markov chain, the process is no longer Markovian.

\subsection{Second left eigenvector}
\label{sec:033eigenvector}

Following the steps of spectral clustering, we now need to obtain an equivalent of the second left eigenvector (Sec.~\ref{sec:023evec_conductance}).
In particular, we now show how to get a relevant eigenvector from the multilinear PageRank vector $\vect{x}$ and the transition tensor $\mat{P}$.
The multilinear PageRank vector $\vect{x}$ satisfying
$\alpha \mat{R} \cdot \left(\vect{x} \otimes \vect{x}\right) + (1 - \alpha)\vect{v} = \vect{x}$
can also be re-interpreted as the stationary distribution of a particular Markov chain.
Specifically, define the matrix 
\begin{equation}
\tenstrans := \sum_{k=1}^{n}x_k\mat{R}_k.
\end{equation}
(Recall that $\mat{R}_k = \tens{P}(:, :, k)$ is the $k$th $n \times n$ block of $\mat{R}$).
The matrix $\tenstrans$ is column stochastic because each $\mat{R}_k$ is column stochastic and $\sum_{k=1}^{n}x_k = 1$.
Note that
\[
\mat{R}\cdot\left(\vect{x} \otimes \vect{x}\right)
 = \sum_{k=1}^{n}\mat{R}_k\left(x_k\vect{x}\right)
 = \left(\sum_{k=1}^{n}x_k\mat{R}_k\right)\vect{x}
 = \tenstrans \cdot \vect{x}.
 \]

Hence, $\vect{x}$ is the stationary distribution of the PageRank system $\alpha\tenstrans \cdot \vect{x} + (1 - \alpha)\vect{v} = \vect{x}$.
However, the transition matrix depends on $\vect{x}$ itself.

We use the second left eigenvector of $\tenstrans$ for our higher-order spectral clustering algorithm.
Heuristically, $\tenstrans$ is a weighted sum of $n$ ``views'' of the graph (the matrices $\mat{R}_k$),
from the perspective of each node ($k$, $1 \le k \le n$),
according to three-dimensional graph data (the tensor $\tens{T}$).
If node $k$ has a large influence on the three-dimensional data, then $x_k$ will be large and we will weight data associated with node $k$ more heavily.
The ordering of the eigenvector will be used for a sweep cut on the vertices.

%
%

\subsection{Sweep cuts}
\label{sec:034sweep}

The last remaining step is to generalize the notion of the sweep cut (Sec.~\ref{sec:024sweep}).
Recall that the sweep cut takes some ordering on the nodes, $\sigma$,
and computes some score $f(S_k)$ for each cut $S_k = \{\sigma_1, \ldots, \sigma_k\}$.
Finally, the sweep cut procedure returns $\arg\max_{S_k} f(S_k)$.
The eigenvector from Sec.~\ref{sec:033eigenvector} provides us with an ordering for a sweep cut,
just as in the two-dimensional case (Sec.~\ref{sec:024sweep}).
We generalize the cut and volume measures as follows:
\begin{eqnarray}
\cutthree{S}  &:=& \sum_{i, j, k \in V}\tens{T}(i, j, k) - \sum_{i, j, k \in S}\tens{T}(i, j, k) - \sum_{i, j, k \in \bar{S}}\tens{T}(i, j, k) \nonumber \\
\volthree{S}  &:=& \sum \tens{T}(S, V, V) \nonumber.
\end{eqnarray}
And we define ``higher-order conductance" (denoted $\phi_3$) and ``higher-order expansion" (denoted $\rho_3$) as
\begin{eqnarray}
\condthree{S} &:=& \frac{\cutthree{S}}{\min\left(\volthree{S}, \volthree{\bar{S}}\right)} \label{eqn:multi_cond} \\
\rho_3(S)     &:=& \frac{\cutthree{S}}{\min\left(|S|, |\bar{S}|\right)}. \label{eqn:multi_exp}
\end{eqnarray}
This definition ensures that $\phi_3(S) \in [0, 1]$, as in standard conductance.

\subsection{Tensor spectral clustering framework}
\label{sec:035algorithm}

\begin{algorithm}[t]
\caption{Tensor Spectral Partitioning}
\KwData{$G = (V, E)$, $|V| = n$, $\tens{T} \in \mathbb{R}^{n \times n \times n}_+$,
                 dangling distribution vector $\vect{u}$,
                $\alpha \in (0, 1)$}
\KwResult{Set of nodes $S \subset V$}
\For{$1 \le i, j, k \le n$, $\tens{T}(i, j, k) \neq 0$}{
   $\tens{P}(i, j, k) \leftarrow \tens{T}(i, j, k) / \sum_{i} \tens{T}(i, j, k)$
}
\For{$j, k$ such that $\sum_{i}\tens{T}(i, j, k) = 0$}{
  $\tens{P}(:, j, k) \leftarrow u$
}
$\vect{x} \leftarrow MultilinearPageRank(\alpha, \tens{P})$ \\
$\mat{R}_k \leftarrow \tens{P}(:, :, k)$ \\
$\tenstrans \leftarrow \sum_{k}x_k\mat{R}_k$ \\
Compute second left eigenvector $\vect{z}$ of $\tenstrans$ \\
$\sigma \leftarrow$ sorted ordering of $\vect{z}$ \\
$S \leftarrow SweepCut(\sigma, G)$
\label{alg:tsp}
\end{algorithm}
\begin{algorithm}[t]
\caption{Tensor Spectral Clustering (\tsc)}
\KwData{$G = (V, E)$, $|V| = n$, $\tens{T} \in \mathbb{R}^{n \times n \times n}_+$,
                 dangling distribution vector $\vect{u}$,
                $\alpha \in (0, 1)$,
                number of clusters $C$}
\KwResult{Partition $\mathcal{P}$ of $V$}
\If{$|\mathcal{P}| < C$}{
  Partition $G$ into $G_1 = (V_1, E_1)$ and $G_2 = (V_2, E_2)$ via Algorithm~\ref{alg:tsp}. \\
  $\mathcal{P} = \mathcal{P} \cup \{ V_1, V_2 \}$. \\
  Recurse on largest component in $\mathcal{P}$.
}
\label{alg:tsc}
\end{algorithm}

We now have higher-order analogs of all the spectral clustering tools from Sec.~\ref{sec:020preliminaries}.
The central routine of our tensor spectral clustering framework is given in Algorithm~\ref{alg:tsp},
which is the tensor spectral partitioning algorithm.
This subroutine takes a data tensor $\tens{T}$ of third-order information about a graph $G$ and partitions the nodes into two sets.
Algorithm~\ref{alg:tsc} is the clustering algorithm that performs recursive bisection in order to decompose the graph into several components.
This algorithm can also be used with other partitioning algorithms \cite{gleich2006hierarchical},
and we will take that approach in Sec.~\ref{sec:050applications}.

\subsection{Complexity}
\label{sec:036complexity}

The complexity of Algorithm~\ref{alg:tsc} depends on the sparsity of the data tensor $\tens{T}$,
\ie the number of higher-order structures in the network.
The algorithm depends on the sparsity in three ways.
First, all of the higher-order structures in the network must be enumerated as an upfront cost.
Second, the sparsity affects the complexity of the multilinear PageRank subroutine in Algorithm~\ref{alg:tsp}.
Third, the number of non-zeroes in $\tenstrans$ is equal to the number of the higher-order structures.
When performing recursive bisection (Algorithm~\ref{alg:tsc}),
there is no upfront cost to enumerate the structures---we only
need to determine which structures are preserved under the partition.

We argue that the upfront cost is not cumbersome.
Triangle enumeration for real-world undirected networks is a well-studied problem~\cite{cohen2009graph,schank2005finding}.
For directed graphs, we can: (1) undirect the graph, (2) use high-performance code to enumerate the triangles, and
(3) stream through the triangles and only keep those that are the directed structure of interest.

Now, we consider the second and third computations.
Let $T$ be the number of non-zeroes in $\tens{T}$.
There are several methods for computing the multilinear PageRank vector in Algorithm~\ref{alg:tsp}~\cite{gleich2014multilinear}.
We use the shifted fixed point method (akin to the symmetric higher-order power method~\cite{kolda2011shifted}).
Each iteration takes $O(T)$ time, and we found that this method converges very quickly---usually within a handful of iterations.
The computation of the second left eigenvector of $\tenstrans$ dominates the running time.
We use the power method to compute this eigenvector.
Since $\tenstrans$ has $T$ non-zero entries, each iteration takes $O(T)$ time.

Finally, we look at the relationship between $T$ and the size of the graph.
In theory, $T$ can be $O(n^3)$, but this is far from what we see in practice.
For the large networks considered in Sec.~\ref{sec:060d3c_large},
$T \le 6m$ (see Table~\ref{tab:net_statistics}).

To summarize, the majority of our time is spent computing the eigenvector of $\tenstrans$.
Each iteration takes $O(T)$ time, and $T \le 6m$ for the algorithms we consider.
Standard spectral algorithms also compute an eigenvector with the power method,
but each iteration is only $O(m)$ time.
Thus, we can think of our algorithms as running within an order of magnitude of standard algorithms.
However, when moving beyond third-order structures, we note that $T$ can be much larger.

\section{Generalizations and directed 3-cycle cut}
\label{sec:040generalizations}
Before transitioning to applications, we mention two important generalities of our framework and discuss directed 3-cycle cuts.
The directed 3-cycle will play an important role for our applications in Sections~\ref{sec:050applications}~and~\ref{sec:060d3c_large}.

\subsection{Generalizations}
\label{sec:041gen}
Our first generalization deals with data beyond three dimensions.
While we have presented the algorithm with three-dimensional data, the same ideas carry through for higher-order data.
The multilinear PageRank vector can still be computed, although $\alpha$ must be smaller to guarantee convergence~\cite{gleich2014multilinear}.
However, in practice, we do not observe large $\alpha$ impeding convergence.

Second, our \tsc algorithm is a strict generalization of traditional spectral clustering in the following sense.
There is a data tensor $\tens{T}$ such that for any multilinear PageRank vector $\vect{x}$, we compute the same eigenvector that conductance cut computes.
In particular, we can always define $\tens{T}(i, j, k) = \mat{A}_{ij}$, where $\mat{A}$ is the adjacency matrix.
Then $\mat{R}_k = \mat{P}$, $1 \le k \le n$, and
$\tenstrans = \sum_{k}x_kP = \mat{P}\sum_{k}x_k = \mat{P}$.

\subsection{Directed 3-cycle cuts}
\label{sec:042d3c}
We now turn our attention to a particular three-dimensional representation of directed graph data: directed 3-cycles (D3Cs),
\ie, sets of edges $(i, j)$, $(j, k)$, and $(k, i)$ for distinct nodes $i$, $j$, and $k$.
Such structures are important for community detection \cite{klymko2014using} and are natural motifs for network feedback.
We will use this structure for applications in Sections~\ref{sec:050applications}~and~\ref{sec:060d3c_large}.
The data tensor we use for directed 3-cycle cuts is
\begin{equation}
   \tens{T}(i, j, k) = \left\{
     \begin{array}{ll}
       2 &  i, j, k \text{ form two D3Cs} \\
       1 &  i, j, k \text{ form one D3C} \\
       0 & \text{otherwise}
     \end{array}
   \right.
   \label{eqn:d3c_tens}
\end{equation}
Nodes $i$, $j$, and $k$ form two D3Cs if and only if every possible directed edge between them is present.
When $\tens{T}(i, j, k) = 1$, we do not differentiate between $0$, $1$, or $2$ reciprocated edges.
For \emph{directed 3-cycle cut}, we want to find partitions of the graph that do not cut many D3Cs.

\subsection{Strongly connected components}
\label{sec:043SCCs}
We now show that \tsc correctly breaks up strongly connected components when using
the data tensor in \eqnref{eqn:d3c_tens}.
Suppose we have an undirected graph $G = (V, E)$ with two connected components $V_1$ and $V_2$.
A standard result of the spectral method for conductance cut on undirected graphs (Sec.~\ref{sec:023evec_conductance})
is that there is a second left eigenvector $\vect{z}$ of $\mat{P}$ such that $\vect{z}^{\Tra}\mat{P} = \vect{z}$,
and $\sign{z_i} = -\sign{z_j}$ for $i \in V_1$, $j \in V_2$ \cite{chung1997spectral}.
This means that the ordering induced by the eigenvector correctly separates the components.
A similar result holds for strongly connected components in a directed graph $G$ using the directed Laplacian.

We now present a similar result for directed 3-cycle cut.
First, we observe the following:
\emph{there is no directed 3-cycle that has nodes from different strongly connected components}.
Now, Lemma~\ref{lem:scc_evec} shows that if we have a graph with two strongly connected components, then,
under some conditions, the second left eigenvector computed by Algorithm~\ref{alg:tsp} correctly partitions the two strongly connected components.

\begin{lemma}\label{lem:scc_evec}
Consider a directed graph $G = (V, E)$ with two components $V_1$ and $V_2$ such that there are no directed 3-cycles containing a node $i \in V_1$ and $j \in V_2$.
Assume that the directed 3-cycle tensor $\tens{T}$ is given by \eqnref{eqn:d3c_tens}.
Augment the corresponding transition matrices $\mat{R}_k$ with a sink node $t$ so that transition involving $j \in V_1$, $k \in V_2$ (or vice versa) jump to the sink node, \ie,
$\tens{P}(i, j, k) = \mathbb{I}\left(i = t\right)$.
Finally, instead of using the dangling distribution vector $\vect{u}$ to fill in $\tens{P}$, assume that when $\sum_{i} \tens{T}(i, j, k) = 0$ for $j, k \in V_1$, $\tens{P}(i, j, k) = \mathbb{I}(i \in V_1) /  |V_1|$.
(And the same for transitions involving $j, k \in V_2$).

Then $\tenstrans$ has a second left eigenvector $\vect{z}$ with eigenvalue $1$ such that
$\vect{z}^{\Tra}\allones = 0$ and $\sign{z_i} = -\sign{z_j}$ for any $i \in V_1$, $j \in V_2$.
\end{lemma}

\begin{proof}
See the full version of the paper.\footnote{Available from \url{https://github.com/arbenson/tensor-sc}.}
\end{proof}

\section{Applications on synthetic networks}
\label{sec:050applications}
We now explore applications of our \tsc framework.
The purpose of this section is to illustrate that explicitly partitioning higher-order network data can improve partitioning and clustering on directed networks.
The examples that follow are small and synthetic but illustrative.
In future work, we plan to use these ideas on real data sets.

For the applications in this section,
we use the following parameters for the tensor spectral clustering algorithm:
$\alpha = 0.99$ for the multilinear PageRank vector, $\gamma = 0.01$ for SS-HOPM,
$\vect{u} = \vect{v} = \frac{1}{n}\allones$,
and the higher-order conductance score function (\eqnref{eqn:multi_cond}).

\subsection{Layered flow networks}
\label{sec:051layered}
Our first example is a network consisting of multiple layers, where feedback loops primarily occur within a layer.
Information tends to flow ``downwards" from one layer to the next.
In other words, most edges between two layers point in the same direction.
Figure~\ref{fig:layered_flow} gives an example of such a network with three layers, each consisting of four nodes.

We are interested in separating the layers of the network via our \tsc algorithm.
Feedback in a directed network is synonymous with directed cycle.
For this example, we count all directed 2-cycles (\ie, reciprocated edges) and directed 3-cycles.
In order to account for the directed 2-cycles, we will say that the data tensor $\tens{T}$ is equal to one for any 
index of the form $(i, i, j)$, $(i, j, i)$, or $(j, i, i)$ when nodes $i$ and $j$ have reciprocated edges.
Formally, the data tensor is:
\[
   \tens{T}(i, j, k) = \left\{
     \begin{array}{ll}
       2 &  i, j, k \text{ distinct and form two D3Cs} \\
       1 &  i, j, k \text{ distinct and form one D3C} \\
       1 & (k = j \text{ or } k = i) \text{ and } (i, j), (j, i) \in E \\
       1 & j = i \text{ and } (i, k), (k, i) \in E    \\
     \end{array}
   \right.
\]

Figure~\ref{fig:layered_flow} lists the three communities found by
(1) \tsc (Algorithm~\ref{alg:tsc} with $C = 3$),
(2) the directed Laplacian (DL), and
(3) the directed Laplacian on the subgraph only including edges involved in at least one directed 2-cycle or directed 3-cycle (Sub-DL).
\tsc is the only method that correctly identifies the three communities.
Sub-DL performs almost as well, but misclassifies node $1$, placing it with the green nodes two layers beneath.
In general, DL does not do well because there are a large number of edges between layers,
and the algorithm does not want to cut these edges.

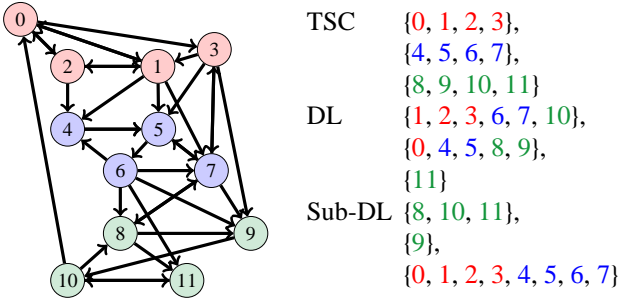
\begin{figure}
\begin{minipage}{\textwidth}
\begin{minipage}{0.22\textwidth}
 \tikzset{first node/.style={circle,fill=red!20,draw,minimum size=0.6cm,inner sep=0pt},}
  \tikzset{second node/.style={circle,fill=blue!20,draw,minimum size=0.6cm,inner sep=0pt},}
  \tikzset{third node/.style={circle,fill=green!20,draw,minimum size=0.6cm,inner sep=0pt},}
\scalebox{0.75}{
 \begin{tikzpicture}
    \node[first node] (0) {$0$};
    \node[first node] (1) [below right = 0.4cm and 2cm of 0]  {$1$};
    \node[first node] (2) [below right = 0.4cm and 0.4cm of 0] {$2$};
    \node[first node] (3) [below right = 0.1cm and 3cm of 0] {$3$};    
    
    \node[second node] (4) [below = 0.5cm of 2] {$4$};
    \node[second node] (5) [right = 1cm of 4] {$5$};    
    \node[second node] (6) [below right = 0.3cm and 0.5cm of 4] {$6$};        
    \node[second node] (7) [below right = 0.3cm and 0.5cm of 5] {$7$};
    
    \node[third node] (8) [below= 0.5cm of 6] {$8$};
    \node[third node] (9) [right = 1.7cm of 8] {$9$};
    \node[third node] (10) [below left = 0.4cm and 0.5cm of 8] {$10$};
    \node[third node] (11) [right = 1.5cm of 10] {$11$};

    \path[draw, ultra thick] (0) edge [->] (1);
    \path[draw, ultra thick] (1) edge [->] (2);
    \path[draw, ultra thick] (2) edge [->] (0);
    \path[draw, ultra thick] (1) edge [->] (0);
    \path[draw, ultra thick] (2) edge [->] (1);
    \path[draw, ultra thick] (0) edge [->] (2);    
    \path[draw, ultra thick] (0) edge [->] (3);        
    \path[draw, ultra thick] (3) edge [->] (1);
    
    \path[draw, ultra thick] (4) edge [->] (5);
    \path[draw, ultra thick] (5) edge [->] (6);
    \path[draw, ultra thick] (6) edge [->] (4);    
    \path[draw, ultra thick] (6) edge [->] (7);        
    \path[draw, ultra thick] (7) edge [->] (5);
    \path[draw, ultra thick] (5) edge [->] (7);            
    
    \path[draw, ultra thick] (8) edge [->] (9);
    \path[draw, ultra thick] (9) edge [->] (10);
    \path[draw, ultra thick] (10) edge [->] (8);    
    \path[draw, ultra thick] (10) edge [->] (11);        
    \path[draw, ultra thick] (11) edge [->] (10);
    \path[draw, ultra thick] (8) edge [->] (11);    
    
    \path[draw, ultra thick] (1) edge [->] (4);    
    \path[draw, ultra thick] (1) edge [->] (5);        
    \path[draw, ultra thick] (1) edge [->] (7);        
    \path[draw, ultra thick] (2) edge [->] (4);    
    \path[draw, ultra thick] (4) edge [->] (5);            
    \path[draw, ultra thick] (3) edge [->] (5);
    \path[draw, ultra thick] (3) edge [->] (7);
    
    \path[draw, ultra thick] (6) edge [->] (8);    
    \path[draw, ultra thick] (6) edge [->] (9);        
    \path[draw, ultra thick] (6) edge [->] (11);      
    \path[draw, ultra thick] (7) edge [->] (8);    
    \path[draw, ultra thick] (7) edge [->] (9);              f
    
    \path[draw, ultra thick] (3) edge [->] (9);            
    
    \path[draw, ultra thick] (10) edge [->] (0); 
    \path[draw, ultra thick] (8) edge [->] (7);     
    \path[draw, ultra thick] (7) edge [->] (3);         
    
\end{tikzpicture}
} 
\end{minipage}
\begin{minipage}{0.7\textwidth}
\begin{tabular}{l @{\hskip 0.15cm} l}
\tsc &
\{\textcolor{red}{0},
\textcolor{red}{1},
\textcolor{red}{2},
\textcolor{red}{3}\}, \\
&\{\textcolor{blue}{4},
\textcolor{blue}{5},
\textcolor{blue}{6},
\textcolor{blue}{7}\},  \\
&\{\textcolor{green}{8},
\textcolor{green}{9},
\textcolor{green}{10},
\textcolor{green}{11}\} \\
DL &
\{\textcolor{red}{1},
\textcolor{red}{2},
\textcolor{red}{3},
\textcolor{blue}{6},
\textcolor{blue}{7},
\textcolor{green}{10}\}, \\
&\{\textcolor{red}{0},
\textcolor{blue}{4},
\textcolor{blue}{5},
\textcolor{green}{8},
\textcolor{green}{9}\}, \\
&\{\textcolor{green}{11}\} \\
Sub-DL&\{\textcolor{green}{8},
\textcolor{green}{10},
\textcolor{green}{11}\}, \\
&\{\textcolor{green}{9}\}, \\
&\{\textcolor{red}{0},
\textcolor{red}{1},
\textcolor{red}{2},
\textcolor{red}{3}, 
\textcolor{blue}{4},
\textcolor{blue}{5},
\textcolor{blue}{6},
\textcolor{blue}{7}\}
\end{tabular}
\end{minipage}
\end{minipage}
\caption{
(Left)
Layered flow network, where almost all feedback occurs at three different layers (specified by the blue, red, and green nodes).
There are many edges going from one layer to the layers below it.
(Right)
Three communities found when using \tsc,
the directed Laplacian (DL),
and the directed Laplacian on the subgraph of
edges participating in at least one directed 2- or 3-cycle (Sub-DL).
Only \tsc correctly identifies all three communities.
\vspace{-0.28cm}
}
\label{fig:layered_flow}
\end{figure}

\subsection{Anomaly detection}
\label{sec:052anomaly}
Our second example is anomaly detection.
In many real networks, most directed 3-cycles have at least one reciprocated edge~\cite{klymko2014using}.
Thus, a set of nodes with many directed 3-cycles and few reciprocated edges between them would be highly anamolous.
The goal of this example is to show that our \tsc framework
can find such sets of nodes when they are planted in a network.

Figure~\ref{fig:anomaly} shows a network where the anomalous cluster we want to identify is nodes 0--5.
All triangles between nodes 0--5 are directed 3-cycles with no reciprocated edges.
Nodes 6--21 connect to each other according to a  Erd\H{o}s-R\'{e}nyi model with edge probability 0.25.
Finally, nodes 0--5 each have four outgoing and two incoming edges with nodes 6--21.
In total, there are 18 directed 3-cycles with no reciprocated edges, and 8 of them occur between nodes 0--5.

To use the \tsc framework, we form a data tensor that only counts directed 3-cycles with no reciprocated edges:
\begin{eqnarray}\label{eqn:d3c_noback}
\tens{T}(i, j, k)
&=& \mathbb{I}\left((i, j), (j, k), (k, i) \in E, (j, i), (k, j), (i, k) \notin E\right) \nonumber \\
&+& \mathbb{I}\left((j, i), (k, j), (i, k) \in E, (i, j), (j, k), (k, i) \notin E\right) \nonumber
\end{eqnarray}

Figure~\ref{fig:anomaly} lists the smaller of the two communities found by
(1) \tsc (Algorithm~\ref{alg:tsc} with $C = 2$),
(2) the directed Laplacian (DL),
and (3) the directed Laplacian on the subgraph only including edges involved in at least one directed 3-cycle with no reciprocated edges (Sub-DL).
We see that only \tsc correctly captures the planted anomalous community.
DL does not capture any information about directed 3-cycles with no reciprocated edges,
and hence the cut does not make sense in this context.
Sub-DL correctly captures nodes $0$, $1$, $4$, and $5$, but misses nodes $2$ and $3$.

\begin{figure}
\begin{minipage}{\textwidth}
\begin{minipage}{0.22\textwidth}

\tikzset{first node/.style={circle,fill=red!20,draw,minimum size=0.6cm,inner sep=0pt},}
\tikzset{second node/.style={circle,fill=white!20,draw,minimum size=0.6cm,inner sep=0pt},}
\scalebox{0.75}{
 \begin{tikzpicture}
    \node[first node] (0) {$0$};
    \node[first node] (1) [above right = 0.2cm and 0.2cm of 0]  {$1$};
    \node[first node] (2) [above left = 0.5cm and 0.5cm of 1]  {$2$};    
    \node[first node] (3) [left = 0.25cm of 2] {$3$};
    \node[first node] (4) [below left = 0.5cm and 0.2cm of 3] {$4$};    
    \node[first node] (5) [below right = 0.25cm and 0.3cm of 4] {$5$};        
    
    \node[second node] (6) [above right = 0.4cm and 0.3cm of 1] {$6$};
    \node[second node] (7) [above = 0.4cm of 6] {$7$};
    \node[second node] (8) [left = 0.25cm of 7] {$8$};
    \node[second node] (9) [left = 0.25cm of 8] {$9$};
    \node[second node] (10) [left = 0.25cm of 9] {$10$};
    \node[second node] (11) [left = 0.25cm of 10] {$11$};
    \node[second node] (12) [below left = 0.05cm and 0.25cm of 11] {$12$};
    \node[second node] (13) [below = 0.25cm of 12] {$13$};
    \node[second node] (14) [below = 0.25cm of 13] {$14$};
    \node[second node] (15) [below = 0.25cm of 14] {$15$};
    \node[second node] (16) [below = 0.25cm of 15] {$16$};
    \node[second node] (17) [right = 0.25cm of 16] {$17$};    
    \node[second node] (18) [right = 0.25cm of 17] {$18$};         
    \node[second node] (19) [right = 0.25cm of 18] {$19$};           
    \node[second node] (20) [right = 0.25cm of 19] {$20$};
    \node[second node] (21) [above right = 0.25cm and 0.1cm of 20] {$21$};
                 
                      \begin{pgfonlayer}{bg}
                      \path[draw, ultra thick] (0) edge [->] (2);
                      \path[draw, ultra thick] (0) edge [->] (3);
                      \path[draw, dashed] (0) edge [->] (10);
                      \path[draw, dashed] (0) edge [->] (12);
                      \path[draw, dashed] (0) edge [->] (15);
                      \path[draw, ultra thick] (1) edge [->] (2);
                      \path[draw, ultra thick] (1) edge [->] (3);
                      \path[draw, dashed] (1) edge [->] (9);
                      \path[draw, dashed] (1) edge [->] (14);
                      \path[draw, dashed] (1) edge [->] (19);
                      \path[draw, dashed] (1) edge [->] (21);
                      \path[draw, ultra thick] (2) edge [->] (4);
                      \path[draw, ultra thick] (2) edge [->] (5);
                      \path[draw, dashed] (2) edge [->] (11);
                      \path[draw, dashed] (2) edge [->] (16);
                      \path[draw, dashed] (2) edge [->] (18);
                      \path[draw, dashed] (2) edge [->] (19);
                      \path[draw, ultra thick] (3) edge [->] (4);
                      \path[draw, ultra thick] (3) edge [->] (5);
                      \path[draw, dashed] (3) edge [->] (8);
                      \path[draw, dashed] (3) edge [->] (16);
                      \path[draw, dashed] (3) edge [->] (18);
                      \path[draw, ultra thick] (4) edge [->] (0);
                      \path[draw, ultra thick] (4) edge [->] (1);
                      \path[draw, dashed] (4) edge [->] (6);
                      \path[draw, dashed] (4) edge [->] (11);
                      \path[draw, dashed] (4) edge [->] (17);
                      \path[draw, dashed] (4) edge [->] (19);
                      \path[draw, ultra thick] (5) edge [->] (0);
                      \path[draw, ultra thick] (5) edge [->] (1);
                      \path[draw, dashed] (5) edge [->] (12);
                      \path[draw, dashed] (5) edge [->] (13);
                      \path[draw, dashed] (5) edge [->] (17);
                      \path[draw, dashed] (5) edge [->] (21);
                      \path[draw, dashed] (6) edge [->] (4);
                      \path[draw, dashed] (6) edge [->] (7);
                      \path[draw, dashed] (6) edge [->] (12);
                      \path[draw, dashed] (6) edge [->] (18);
                      \path[draw, dashed] (6) edge [->] (20);
                      \path[draw, dashed] (7) edge [->] (6);
                      \path[draw, dashed] (7) edge [->] (14);
                      \path[draw, dashed] (7) edge [->] (16);
                      \path[draw, dashed] (7) edge [->] (21);
                      \path[draw, dashed] (8) edge [->] (3);
                      \path[draw, dashed] (8) edge [->] (15);
                      \path[draw, dashed] (8) edge [->] (17);
                      \path[draw, dashed] (8) edge [->] (19);
                      \path[draw, dashed] (8) edge [->] (20);
                      \path[draw, dashed] (9) edge [->] (1);
                      \path[draw, dashed] (9) edge [->] (5);
                      \path[draw, dashed] (9) edge [->] (8);
                      \path[draw, dashed] (9) edge [->] (15);
                      \path[draw, dashed] (10) edge [->] (20);
                      \path[draw, dashed] (11) edge [->] (4);
                      \path[draw, dashed] (11) edge [->] (5);
                      \path[draw, dashed] (11) edge [->] (6);
                      \path[draw, dashed] (11) edge [->] (7);
                      \path[draw, dashed] (11) edge [->] (10);
                      \path[draw, dashed] (11) edge [->] (15);
                      \path[draw, dashed] (12) edge [->] (3);
                      \path[draw, dashed] (12) edge [->] (9);
                      \path[draw, dashed] (12) edge [->] (13);
                      \path[draw, dashed] (12) edge [->] (15);
                      \path[draw, dashed] (12) edge [->] (18);
                      \path[draw, dashed] (12) edge [->] (19);
                      \path[draw, dashed] (13) edge [->] (0);
                      \path[draw, dashed] (13) edge [->] (9);
                      \path[draw, dashed] (13) edge [->] (11);
                      \path[draw, dashed] (13) edge [->] (12);
                      \path[draw, dashed] (13) edge [->] (16);
                      \path[draw, dashed] (13) edge [->] (18);
                      \path[draw, dashed] (14) edge [->] (7);
                      \path[draw, dashed] (14) edge [->] (13);
                      \path[draw, dashed] (14) edge [->] (20);
                      \path[draw, dashed] (15) edge [->] (0);
                      \path[draw, dashed] (15) edge [->] (10);
                      \path[draw, dashed] (15) edge [->] (14);
                      \path[draw, dashed] (15) edge [->] (18);
                      \path[draw, dashed] (15) edge [->] (20);
                      \path[draw, dashed] (16) edge [->] (1);
                      \path[draw, dashed] (16) edge [->] (8);
                      \path[draw, dashed] (16) edge [->] (12);
                      \path[draw, dashed] (16) edge [->] (13);
                      \path[draw, dashed] (16) edge [->] (15);
                      \path[draw, dashed] (16) edge [->] (20);
                      \path[draw, dashed] (17) edge [->] (6);
                      \path[draw, dashed] (17) edge [->] (8);
                      \path[draw, dashed] (17) edge [->] (16);
                      \path[draw, dashed] (18) edge [->] (2);
                      \path[draw, dashed] (18) edge [->] (9);
                      \path[draw, dashed] (18) edge [->] (11);
                      \path[draw, dashed] (18) edge [->] (16);
                      \path[draw, dashed] (18) edge [->] (20);
                      \path[draw, dashed] (19) edge [->] (8);
                      \path[draw, dashed] (19) edge [->] (18);
                      \path[draw, dashed] (19) edge [->] (20);
                      \path[draw, dashed] (20) edge [->] (9);
                      \path[draw, dashed] (20) edge [->] (13);
                      \path[draw, dashed] (20) edge [->] (18);
                      \path[draw, dashed] (21) edge [->] (7);
                      \path[draw, dashed] (21) edge [->] (14);
                      \path[draw, dashed] (21) edge [->] (15);
                      \path[draw, dashed] (21) edge [->] (17);
                      \path[draw, dashed] (21) edge [->] (19);
                      \path[draw, dashed] (21) edge [->] (20);
                           \end{pgfonlayer}
\end{tikzpicture}
}
\end{minipage}
\begin{minipage}{0.2\textwidth}
\begin{tabular}{l @{\hskip 0.15cm} l}
\tsc & \{\textcolor{red}{0}, \textcolor{red}{1}, \textcolor{red}{2}, \textcolor{red}{3}, \textcolor{red}{4}, \textcolor{red}{5},\\
& \phantom{\{}12, 13, 16\} \\
DL & \{\textcolor{red}{1}, \textcolor{red}{4}, \textcolor{red}{5}, 7, 8, 12, \\
& \phantom{\{}13, 15, 18, 20\} \\
Sub-DL & \{\textcolor{red}{0}, \textcolor{red}{1}, \textcolor{red}{4}, \textcolor{red}{5}, 9, 11\\
& \phantom{\{}16, 17, 19, 20\} \\
\end{tabular}
\end{minipage}
\end{minipage}

\caption{
(Left)
Network with planted anomalous cluster (nodes 0--5).
Between these nodes, there are many directed 3-cycles with no reciprocated edges (thick black lines).
Nodes 6--22 follow an Erd\H{o}s-R\'{e}nyi graph pattern with edges indicated by dashed lines.
(Right)
Smaller of two communities found by \tsc,
the directed Laplacian (DL),
and the directed Laplacian on the subgraph with only edges involved in a directed 3-cycle with no reciprocated edges (Sub-DL).
Only \tsc finds the entire anomalous cluster.
\vspace{-0.28cm}
}
\label{fig:anomaly}
\end{figure}
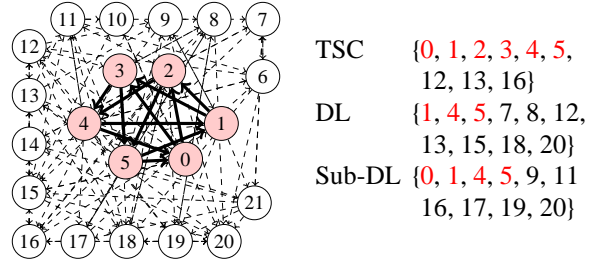

\section{Directed 3-cycle cuts on large networks}
\label{sec:060d3c_large}
We now transition to real data sets and show that our tensor spectral partitioning algorithm provides good cuts
for the directed 3-cycle (D3C) data tensor given by \eqnref{eqn:d3c_tens}.
We compare the following algorithms:
\vspace{-0.05cm}
\begin{itemize}\setlength{\itemsep}{-0.15cm}
\item \textbf{TSC}:
This is our proposed method (Algorithm~\ref{alg:tsc} with $C = 2$),
where the data tensor is given by \eqnref{eqn:d3c_tens}.
The sweep cut ordering is provided by the second left eigenvector of $\tenstrans$.
\item \textbf{Undirected Laplacian (UL)}:
The sweep cut ordering is provided by the second left eigenvector of the transition matrix of the
undirected version of the graph.
\item \textbf{Directed Laplacian (DL)} \cite{chung2005laplacians}:
The sweep cut ordering is provided by the second left eigenvector of $\mat{P}_{sym}$ in \eqnref{eqn:P_sym}.
\item \textbf{Asymmetric Laplacian (AL)} \cite{boley2011commute}:
The sweep cut ordering is provided by the second left eigenvector of $\mat{P}$.
\item \textbf{Co-clustering (Co)} \cite{dhillon2001co, rohe2012co}:
The sweep cut ordering is based on the second left and right singular vectors of a normalized adjacency matrix.
Specifically, let
$\mat{D}_{row} = \text{diag}(\mat{A}\allones)$ and 
$\mat{D}_{col} = \text{diag}(\mat{A}^{\Tra}\allones)$ and
let  $\mat{U}\mat{\Sigma}\mat{V}^{\Tra}$  be the singular value decomposition of $\mat{D}_{row}^{-1/2}\mat{A}\mat{D}_{col}^{-1/2}$.
The the sweep cut ordering is provided by $\mat{D}_{row}^{-1/2}\mat{U}(:, 2)$ or $\mat{D}_{col}^{-1/2}\mat{V}(:, 2)$.
We take the better of the two cuts.
\item \textbf{Random}:
The sweep cut ordering is random.
This provides a simple baseline.
\end{itemize}

\subsection{Data preprocessing}
\label{sec:061data_preprocessing}
Before running partitioning algorithms, we first filter the networks as follows:
(1) remove all edges that do not participate in any D3C, and
(2) take the largest strongly connected component of the remaining network.
We perform this filtering to make a fairer comparison between the different partitioning algorithms.
Table~\ref{tab:net_statistics} lists the relevant networks and statistics for the filtered networks that we use in our experiments.
We limit ourselves to a few representative networks to illustrate the main patterns we observed.
Data for more networks is available in the full version of this paper.
Networks are available from SNAP~\cite{snapnets}.

\begin{table}[tb]
\centering
\caption{
Statistics of networks used for computing directed 3-cycle cuts.
The statistics are taken on the largest strongly connected component of the network
after removing all edges that do not participate in any D3C.
}
\begin{tabular}{l l l l l}
\toprule
Network & $n = |V|$ & $m = |E|$ & \# D3Cs \\ \midrule
\dataset{email-EuAll} & 11,315 & 80,211 & 183,836 \\
\dataset{soc-Epinions1} & 15,963 & 262,779 & 738,231 \\
\dataset{wiki-Talk} & 52,411 & 957,753 & 5,138,613 \\
\dataset{twitter\_combined} & 57,959 & 1,371,621 & 6,921,399 \\
\bottomrule
\vspace{-0.9cm}
\end{tabular}
\label{tab:net_statistics}
\end{table}

\subsection{Results}
\label{sec:062results}
\begin{figure*}[tb]
\centering
\includegraphics[height=3.0cm]{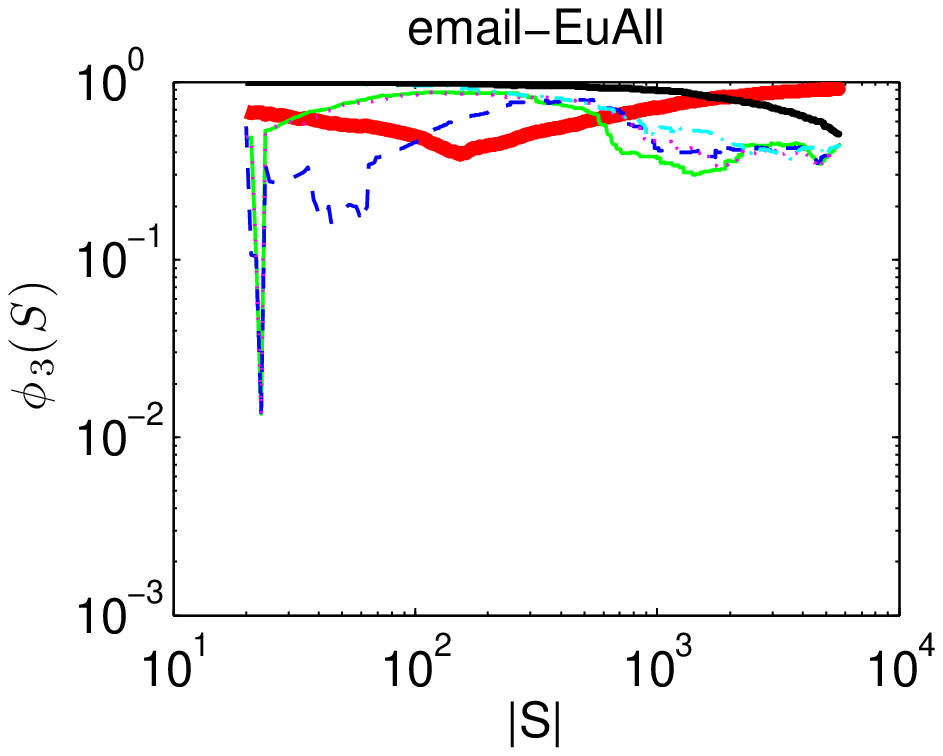}
\includegraphics[height=3.0cm]{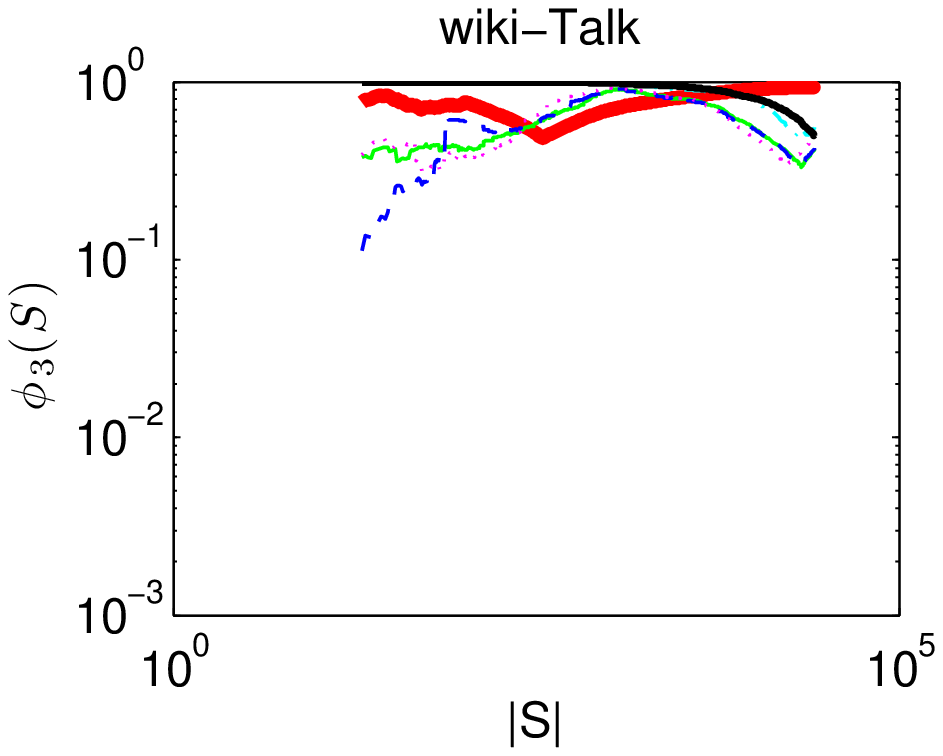}
\includegraphics[height=3.0cm]{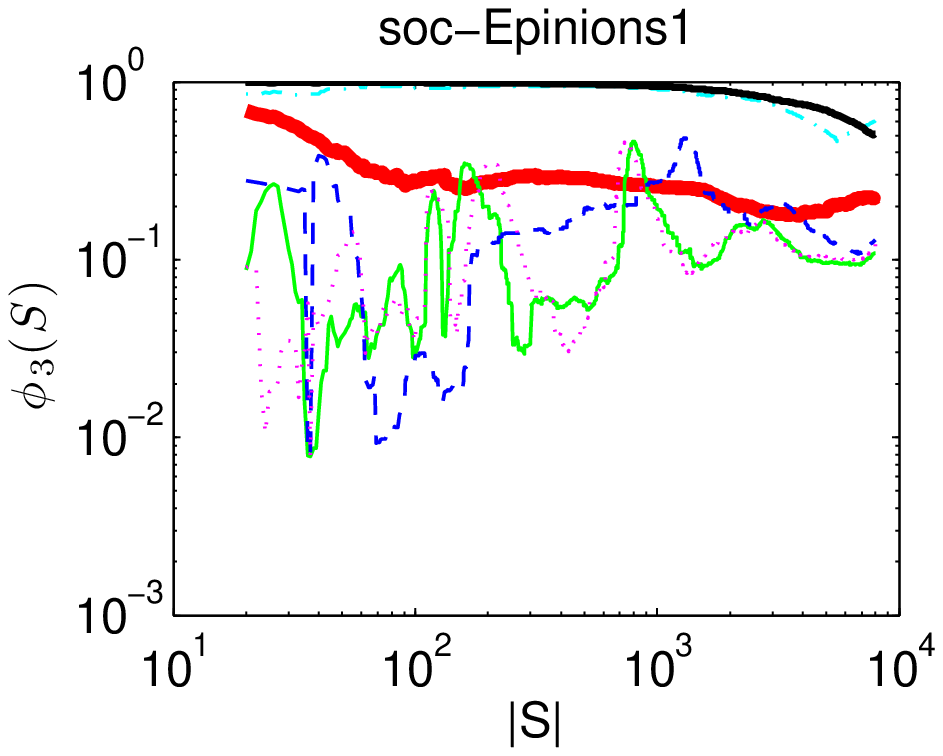}
\includegraphics[height=3.0cm]{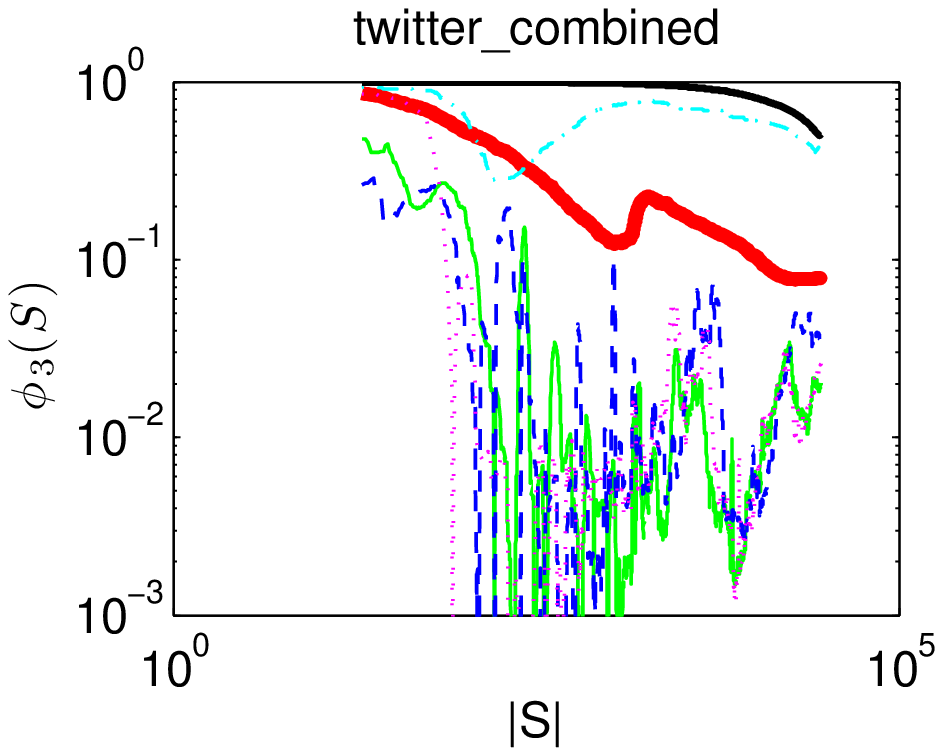} \\\vspace{-0.05cm}
\includegraphics[height=3.0cm]{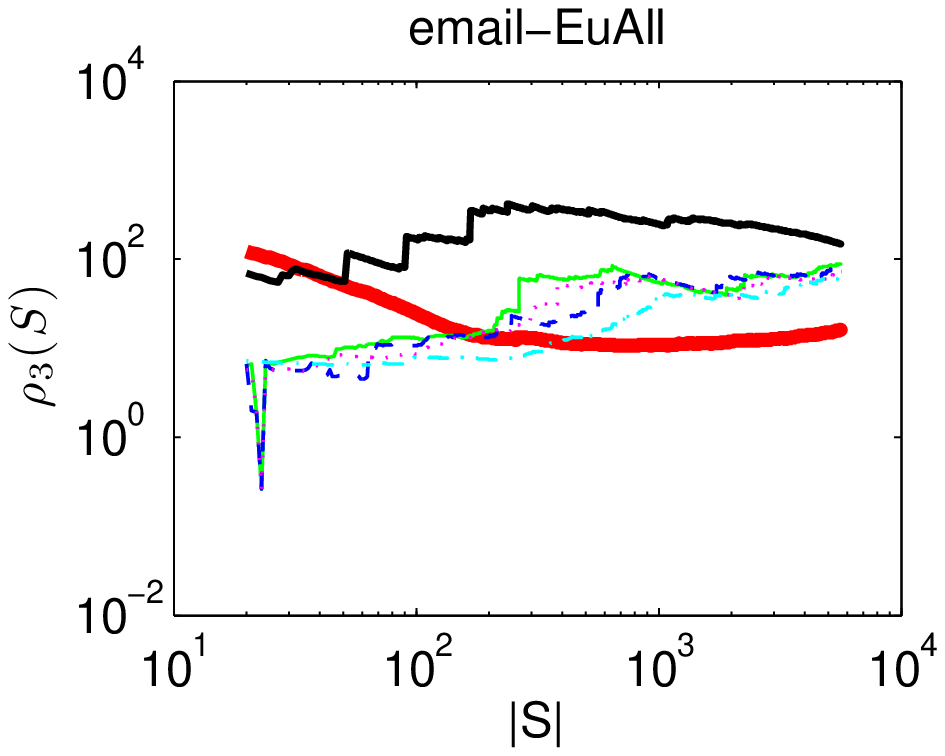}
\includegraphics[height=3.0cm]{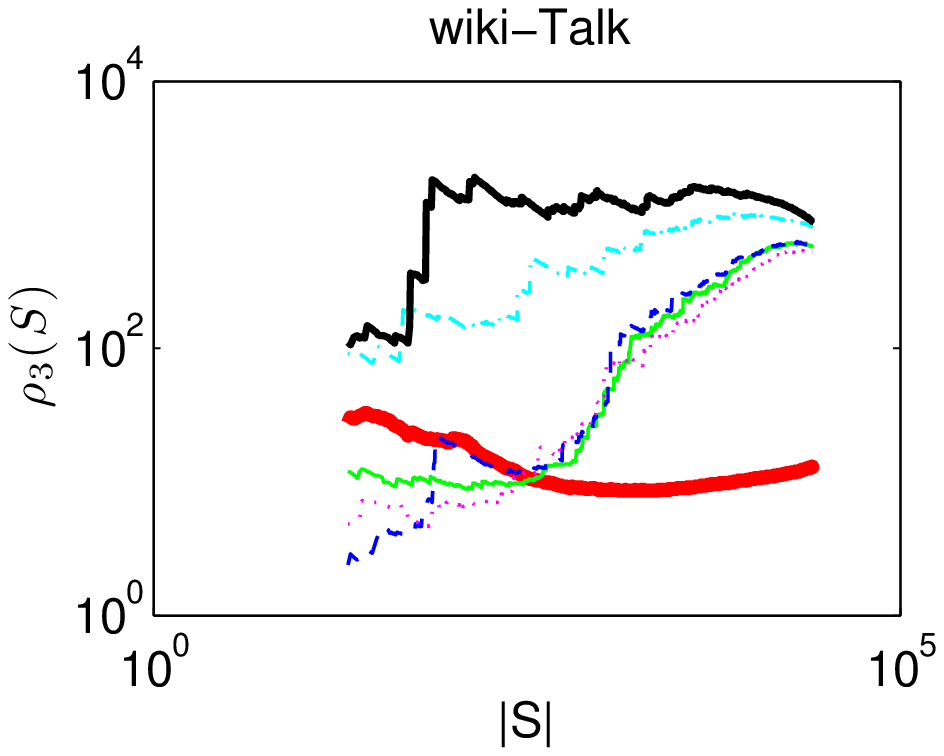}
\includegraphics[height=3.0cm]{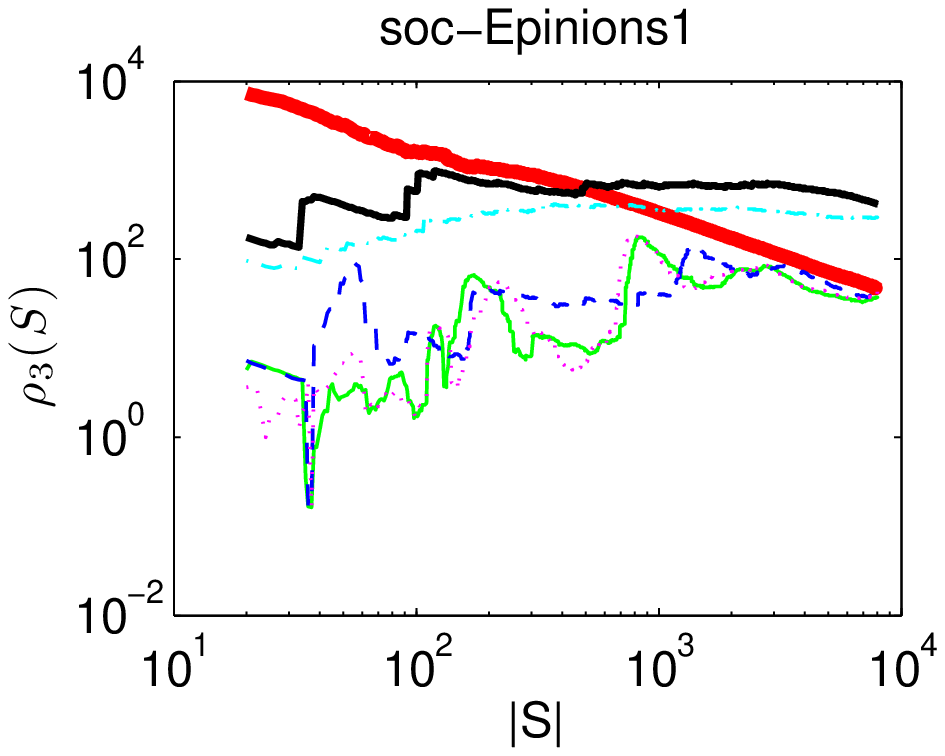}
\includegraphics[height=3.0cm]{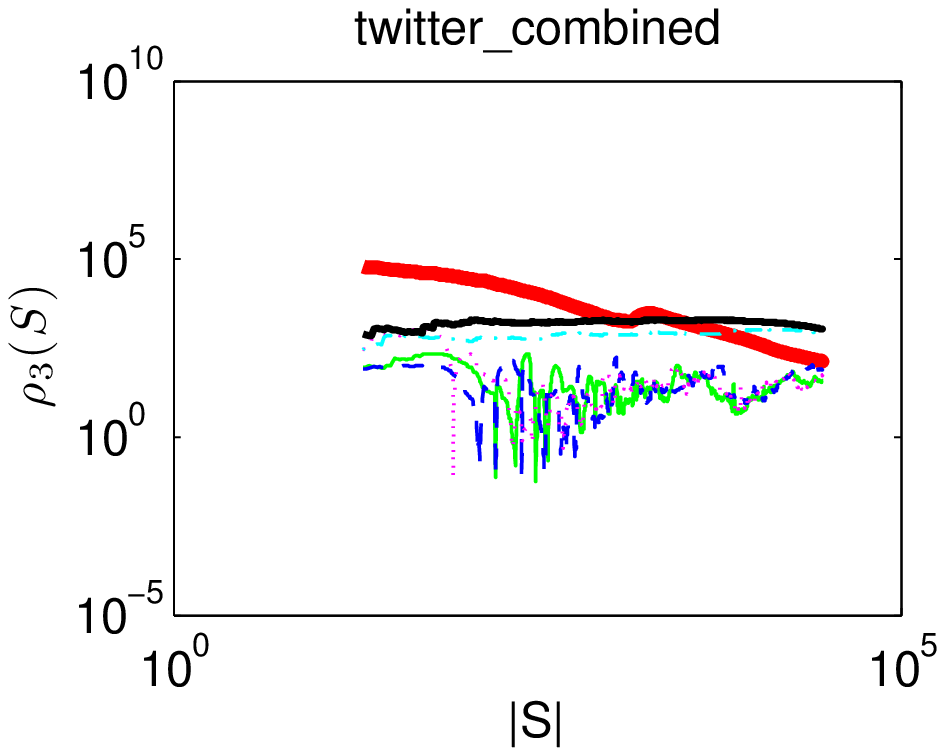} \\\vspace{-0.05cm}
\includegraphics[height=3.0cm]{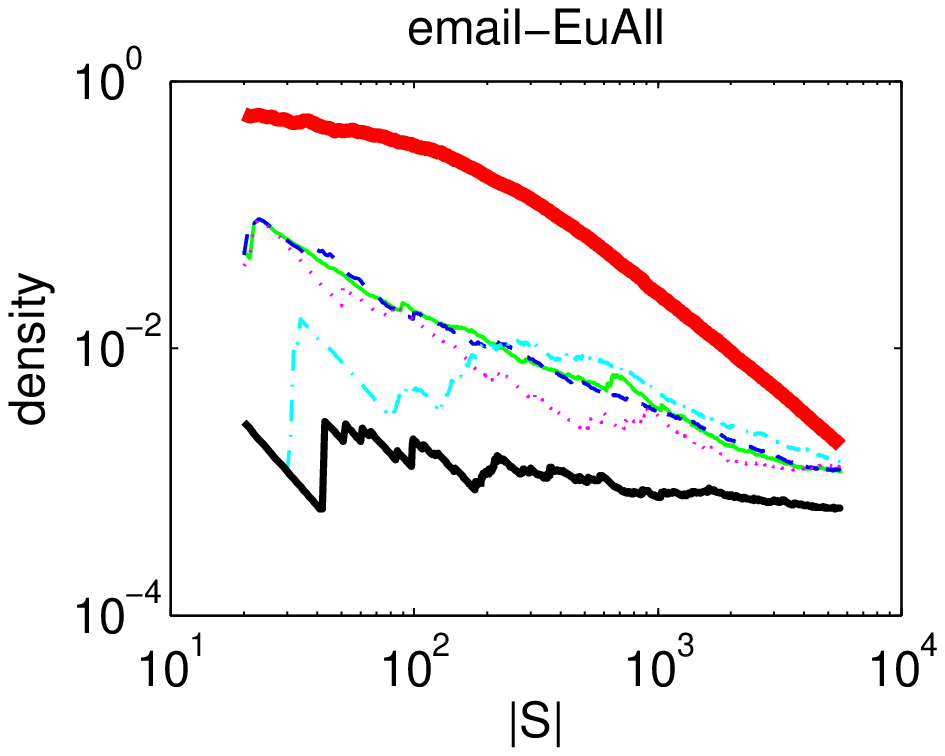}
\includegraphics[height=3.0cm]{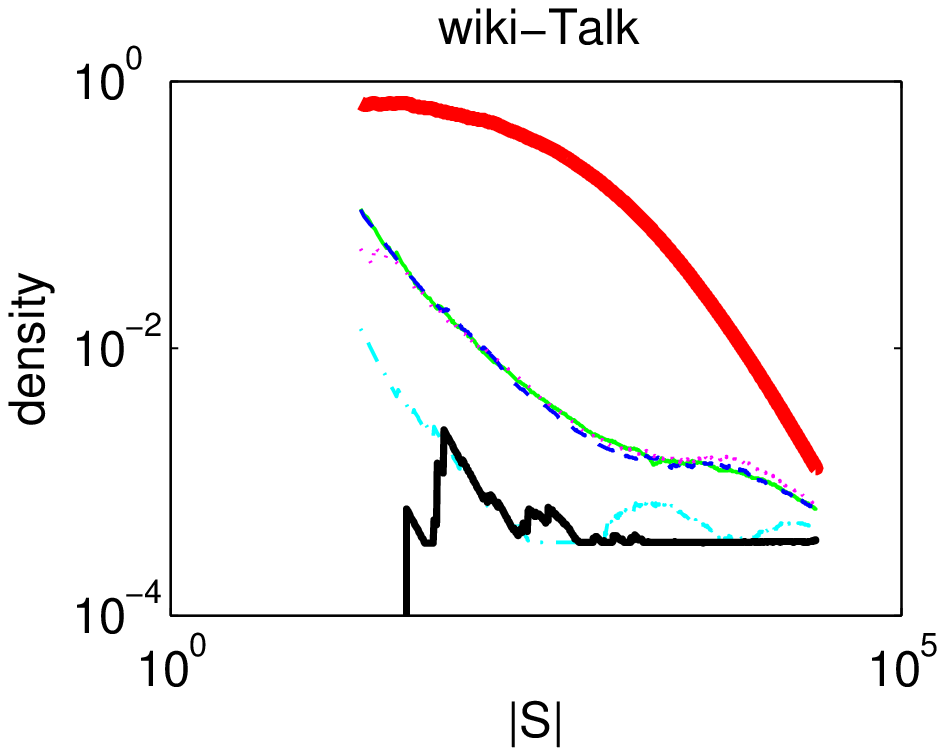}
\includegraphics[height=3.0cm]{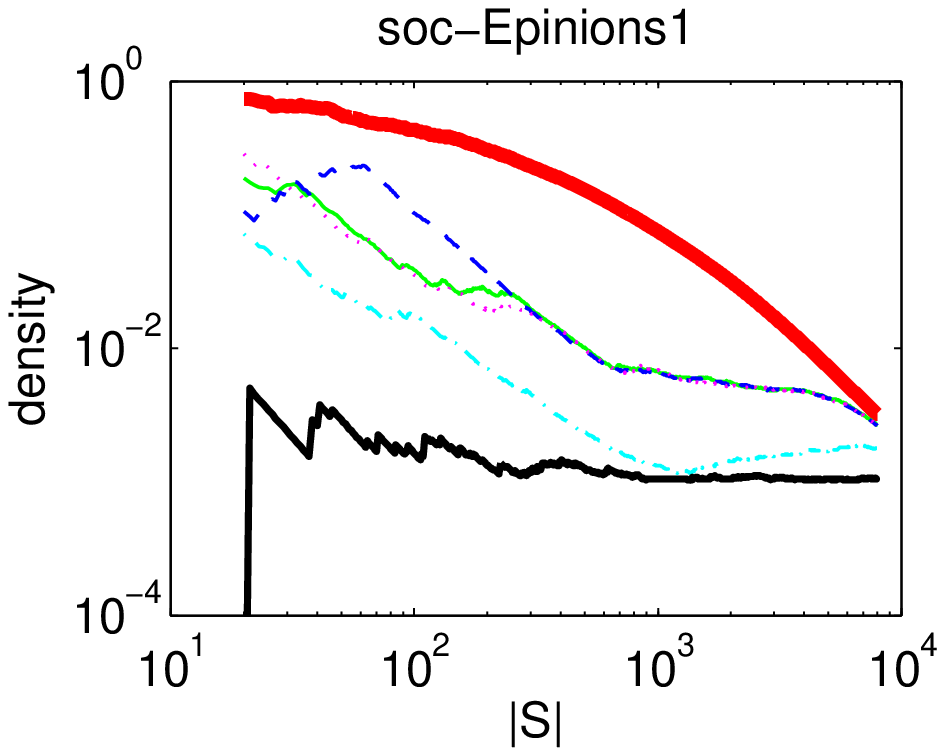}
\includegraphics[height=3.0cm]{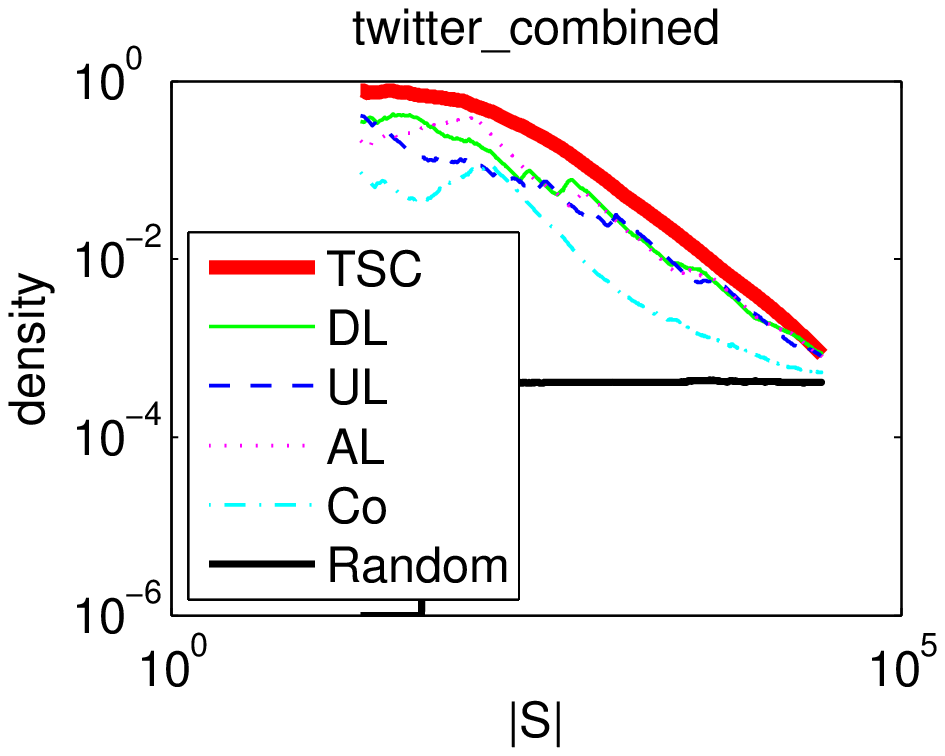}
\vspace{-0.45cm}
\caption{
(Top row)
Higher-order conductance, $\phi_3(S)$, as a 
function of the smaller vertex partition set ($|S|$).
The size of the vertex set runs from twenty to half the nodes in the network.
(Middle row)
Higher-order expansion, $\rho_3(S)$.
(Bottom row)
Density of the cluster.
For \dataset{email-EuAll} and \dataset{wiki-Talk} (left two coumns),
higher-order conductance from \tsc is on par with other spectral methods,
and higher-order expansion is better for large enough clusters.
For \dataset{soc-Epinions1} and \dataset{twitter\_combined},
the higher-order conductance and expansion is better using standard clustering algorithms.
In all cases, \tsc finds much denser clusters.
}
\label{fig:cut_results}
\vspace{-7mm}
\end{figure*}

Figure~\ref{fig:cut_results} shows the sweep profiles
on the networks in Table~\ref{tab:net_statistics}.
The results are for a single cut of the network.
The plots show the
higher-order conductance (\eqnref{eqn:multi_cond}),
higher-order expansion (\eqnref{eqn:multi_exp}),
and density of the smaller of the partitionined vertex sets.
For \dataset{email-EuAll} and \dataset{wiki-Talk},
higher-order conductance is the same for most algorithms,
but \tsc has much better higher-order expansion when the vertex set gets large enough.
On \dataset{soc-Epinions1} and \dataset{twitter\_combined},
standard spectral methods have better higher-order conductance and higher-order expansion.
Crucially, in all cases, \tsc finds much denser subgraphs.
In general, we expect communities with lots of directed 3-cycles to be dense sets.
Thus, even though \tsc sometimes does not always do well  with respect to
the score metrics discussed in Sec.~\ref{sec:034sweep},
it is still finding relevant structure.

Since our goal is to explore structural properties of the cuts,
we did not tune our TSC algorithms for high performance.
Subsequently, we do not compare running times of the algorithms.
However, we note that for each network, our straightforward implementation of TSC
ran in under 10 minutes using a laptop with 4GB of memory.

\section{Related work}
\label{sec:070related}
While the bulk of community detection algorithms are for undirected networks,
there is still an abundance of methods for directed networks \cite{malliaros2013clustering}.
There are several spectral algorithms related to partitioning directed networks.
The ones we investigated in this paper were based on
the undirected Laplacian (\ie, standard spectral clustering but ignoring edge directions),
the directed Laplacian~\cite{chung2005laplacians},
the asymmetric Laplacian~\cite{boley2011commute},
and co-clustering~\cite{dhillon2001co, rohe2012co}.
Other spectral algorithms are based on dyadic methods~\cite{li2012mutual} and optimizing directed modularity~\cite{leicht2008community}.

There is some work in community detection that explicitly targets higher-order structures.
Klymko \emph{et al}.~weight directed edges in triangles and then revert to a clustering algorithm for undirected networks~\cite{klymko2014using}.
Clique percolation builds overlapping communities by examining small cliques~\cite{derenyi2005clique}.
Optimizing the LinkRank metric can identify communities based on information flow~\cite{kim2010finding}, which is similar to our use of directed 3-cycles in Sec.~\ref{sec:051layered}. 
Multi-way relationships between nodes are also explicitly handled by hypergraph partitioners~\cite{karypis2000multilevel} .

Finally, we mention that tensor factorizations have been used by Huang \emph{et al}.~to find overlapping communities \cite{huang2013fast}.
This work uses new spectral techniques for learning latent variable models from higher-order moments \cite{anandkumar2012tensor}.

\section{Discussion}
\label{sec:080discussion}
We have provided a framework for tensor spectral clustering that opens the door to further higher-order network analysis.
The framework gives the user the flexibility to cluster structures based on his or her application.
In Sec.~\ref{sec:050applications} we provided two applications---layered flow networks and anomaly detection---that showed
how this framework can lead to better clustering of nodes based on network motifs.
For these applications, the networks were small and manually constructed.
In future work, we plan to explore these applications on large networks.

In Sec.~\ref{sec:060d3c_large}, we explored clustering based on directed 3-cycles.
In some cases, TSC provided much better cuts in terms of higher-order expansion.
Interestingly, for some networks, simply removing edges that do not participate in a directed 3-cycle 
and using a standard spectral clustering algorithm is sufficient for finding good cuts with respect to higher-order conductance and higher-order expansion.
However, in these cases, we are comparing against baselines optimized for our specific problem.
That being said, \tsc does always identify much denser clusters.
The networks we analyzed were social and internet-based, and it would be interesting to see if similar trends hold for networks derived from physical or biological systems.

For the large networks, we did not perform full directed clustering---we only investigated the sweep profiles.
The higher-level goal of this paper is to explore the ideas in higher-order clustering,
and we leave full-stack algorithms to future work.
One interesting question for such algorithms is whether we should partition based on recursive bisection (Algorithm~\ref{alg:tsc}) or k-means.
These algorithmic variations provide several opportunities for challenging future work.

\section*{Acknowledgements} 
This research has been supported in part by NSF
IIS-1016909,
CNS-1010921,
CAREER IIS-1149837,
ARO MURI,                 
DARPA XDATA,            
DARPA GRAPHS,           
Boeing,  
Facebook,
Volkswagen,
and Yahoo.
David F.~Gleich is supported by NSF CAREER CCF-1149756 and IIS-1422918.
Austin R.~Benson is supported by a Stanford Graduate Fellowship.

\clearpage
\vspace{-0.25cm}
\bibliographystyle{abbrv}
\bibliography{bibliography}


\section{Proof of Lemma~\ref{lem:scc_evec}}
\label{sec:0A0proof}

\begin{proof}
Without loss of generality, order the nodes so that the nodes in $V_1$ are first, the nodes $V_2$ are second, and the sink node $t$ is last.
Let $R_k^{(i)}$ be the transition probabilities of $R_k$, restricted to nodes in $V_i$, $i = 1, 2$.
Because nodes from two different strongly connected components cannot be in the same directed 3-cycle,
\[
\begin{bmatrix}
R_k = \mat{R}_k^{(1)} & 0 & 0 \\
0 & 0 & 0 \\
0 & \allones^{\Tra} & 1 \\
\end{bmatrix}, k \in V_1; \quad
 \mat{R}_k = \begin{bmatrix}
0 & 0 & 0 \\
0 & \mat{R}_k^{(2)} & 0 \\
\allones^{\Tra} & 0 & 1 \\
\end{bmatrix}, k \in V_2
\]
Here, the first block diagonal matrix is of size $|V_1| \times |V_1|$, the second of size $|V_2| \times |V_2|$, and the third of size $1 \times 1$.
Consider the following vectors that have the same block structure:
\[
\vect{y}^{\Tra}_1 = \begin{bmatrix} \allones^{\Tra} & 0    & 1 \end{bmatrix}, \quad
\vect{y}^{\Tra}_2 = \begin{bmatrix} 0    & \allones^{\Tra} & 1 \end{bmatrix}.
\]
Then $\vect{y}_{i}^{\Tra}\mat{R}_k = \vect{y}_i$, $i = 1, 2$ for $k \in V_1$ \emph{and} $k \in V_2$.
Thus,
\[
\vect{y}_i^{\Tra}\tenstrans = \sum_{k}\vect{y}_i^{\Tra}\left(x_k\mat{R}_k\right) = \sum_{k}x_k\vect{y}^{\Tra}_i = \vect{y}_i^{\Tra}
\]

The vector $\vect{z} = \vect{y}_1 - \frac{N_1 + 1}{N_2 + 1}y_2$ satisfies $\vect{z}^{\Tra}\tenstrans = \vect{z}^{\Tra}$ and $\vect{z}^{\Tra}\allones = 0$.
Since $\allones^{\Tra}\tenstrans = \allones^{\Tra}$, $\vect{z}$ is a second left eigenvector.
Finally, $z_i$ is positive for all $i \in V_1$ and negative for all $i \in V_2$.
\end{proof}

\clearpage
\section{Directed 3-cycle cuts on more networks}
\label{sec:0B0more_nets}
Here we present the results of Sec.~\ref{sec:060d3c_large} on more networks.
Table~\ref{tab:net_statistics_all} lists the statistics of eleven networks that we consider.
We include one undirected network, \dataset{email-Enron}.
For this data set, all undirected edges are simply replaced with two directed edges.
Figures~\ref{fig:all_d3c_cond},~\ref{fig:all_d3c_exp},~and~\ref{fig:all_density} show
the sweep profiles for higher-order conductance, higher-order expansion, and density, respectively.

\begin{table}[h]
\centering
\caption{
Statistics of networks used for computing directed 3-cycle cuts.
The statistics are taken on the largest strongly connected component of the network
after removing all edges that do not participate in any directed 3-cycle.
}
\begin{tabular}{l l l l l}
\toprule
Network & $n = |V|$ & $m = |E|$ & \# D3Cs \\ \midrule
\dataset{wiki-Vote} & 1,151 & 24,349 & 43,975 \\
\dataset{wiki-RfA} & 2,219 & 61,965 & 133,004 \\
\dataset{as-caida20071105} & 8,320 & 50,016 & 72,664 \\
\dataset{email-EuAll} & 11,315 & 80,211 & 183,836 \\
\dataset{web-Stanford} & 12,520 & 105,376 & 212,639 \\
\dataset{soc-Epinions1} & 15,963 & 262,779 & 738,231 \\
\dataset{soc-Slashdot0811} & 22,193 & 377,172 & 883,884 \\
\dataset{email-Enron} & 22,489 & 332,396 & 1,447,534 \\
\dataset{wiki-Talk} & 52,411 & 957,753 & 5,138,613 \\
\dataset{twitter\_combined} & 57,959 & 1,371,621 & 6,921,399 \\
\dataset{amazon0312} & 253,405 & 1,476,377 & 1,682,909 \\
\bottomrule
\end{tabular}
\label{tab:net_statistics_all}
\end{table}

\begin{figure*}
\includegraphics[height=3.4cm]{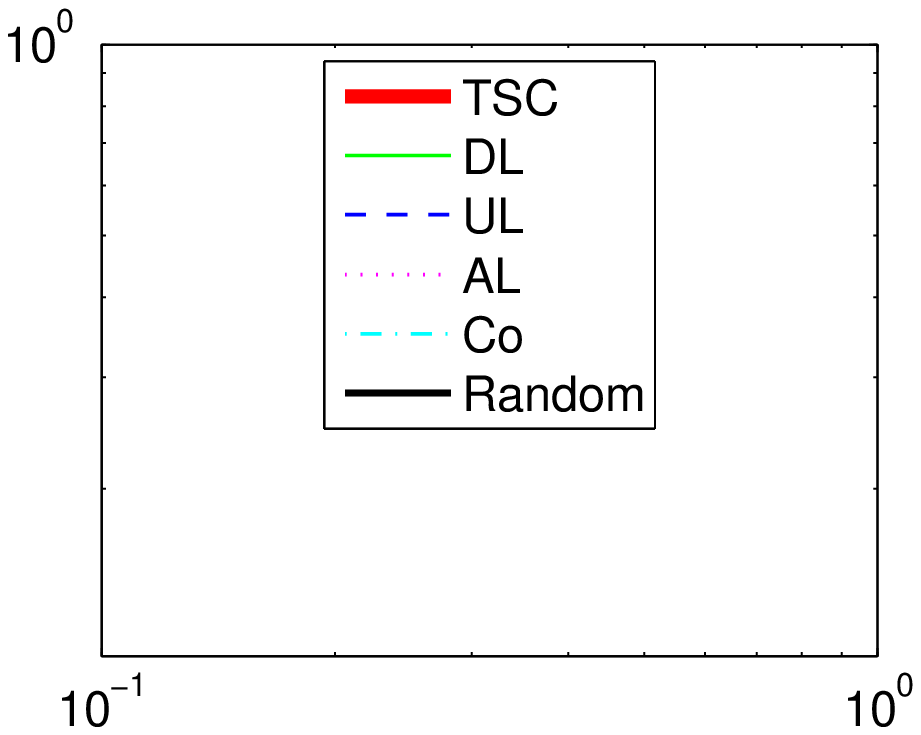}
\includegraphics[height=3.4cm]{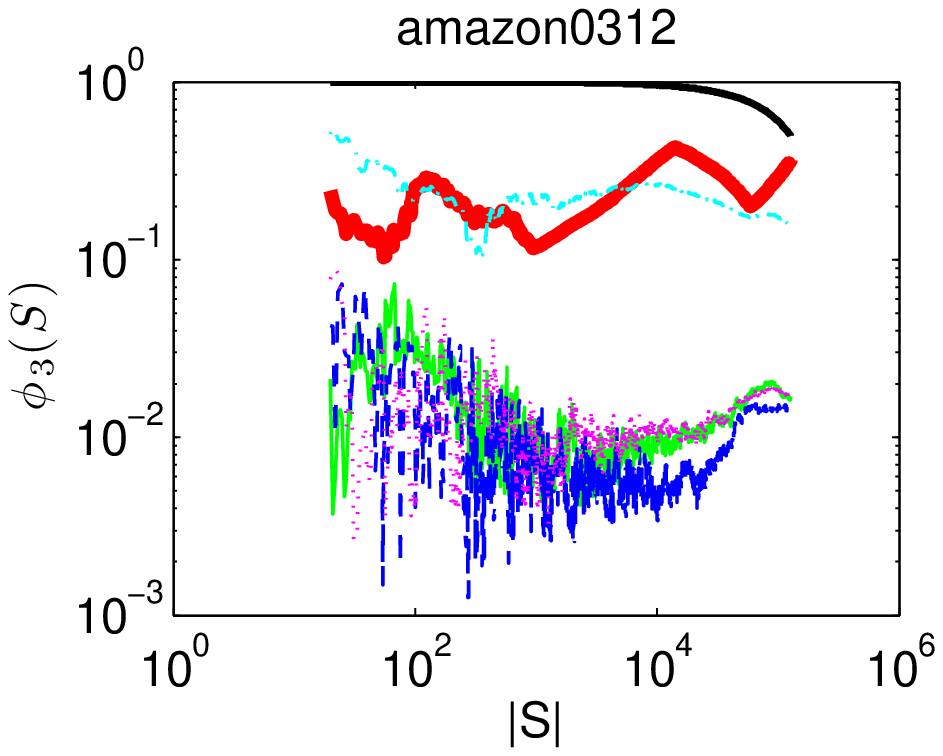}
\includegraphics[height=3.4cm]{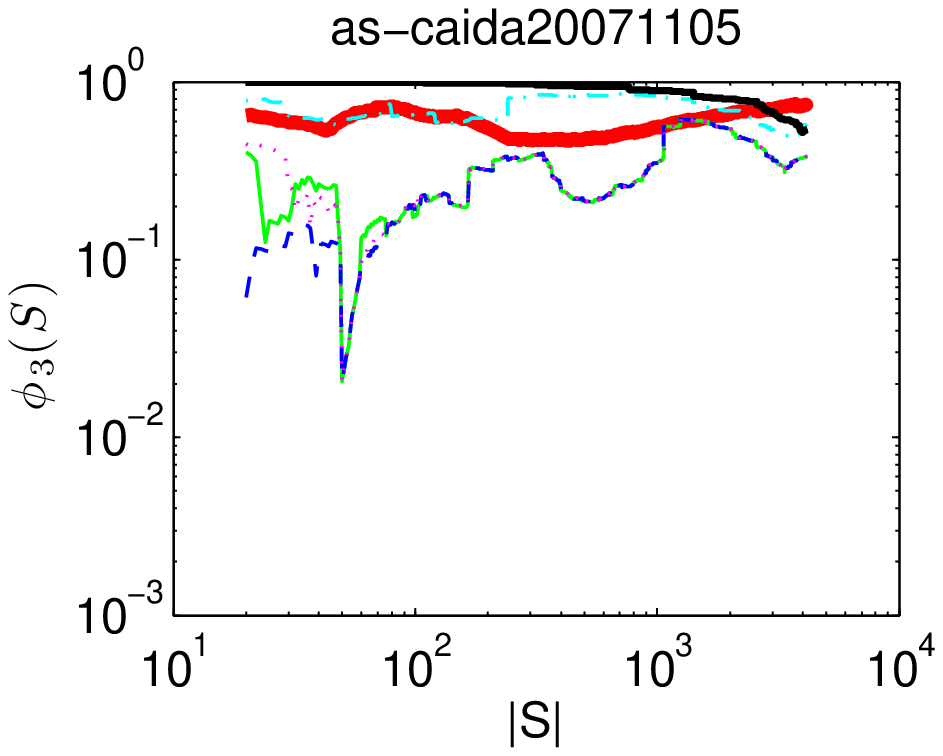}
\includegraphics[height=3.4cm]{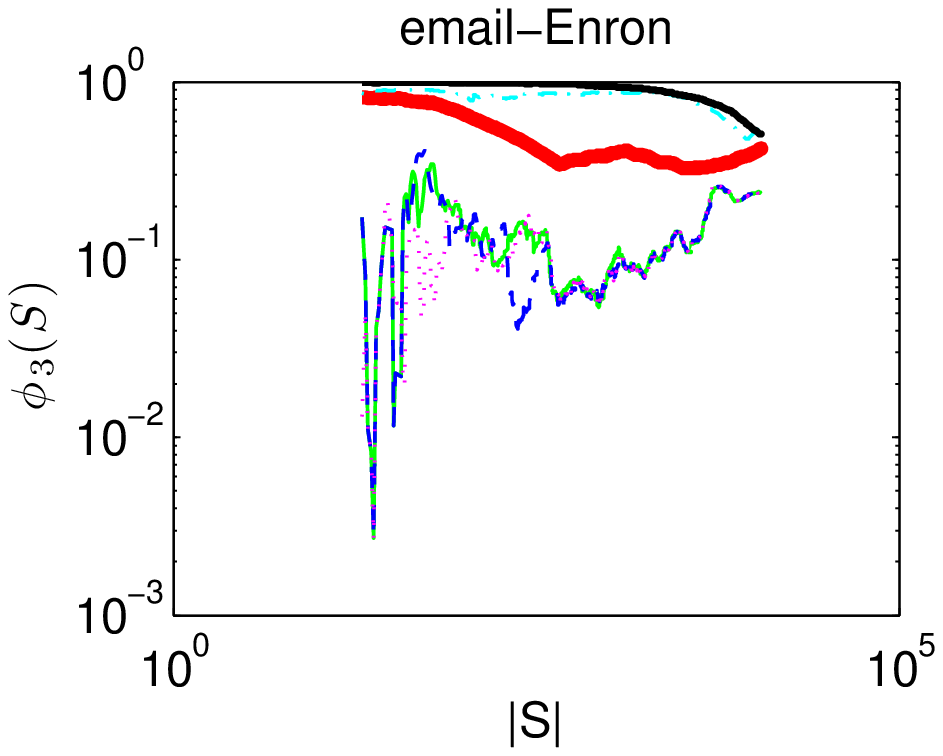} \\
\includegraphics[height=3.4cm]{email-EuAll.eps}
\includegraphics[height=3.4cm]{soc-Epinions1.eps}
\includegraphics[height=3.4cm]{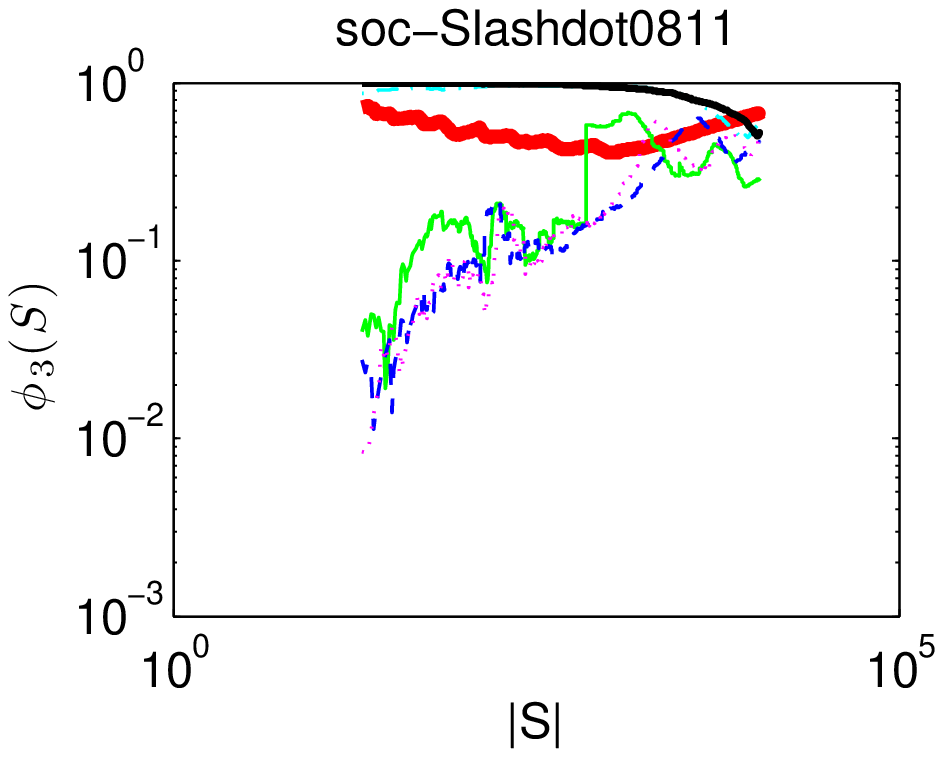}
\includegraphics[height=3.4cm]{twitter_combined.eps} \\
\includegraphics[height=3.4cm]{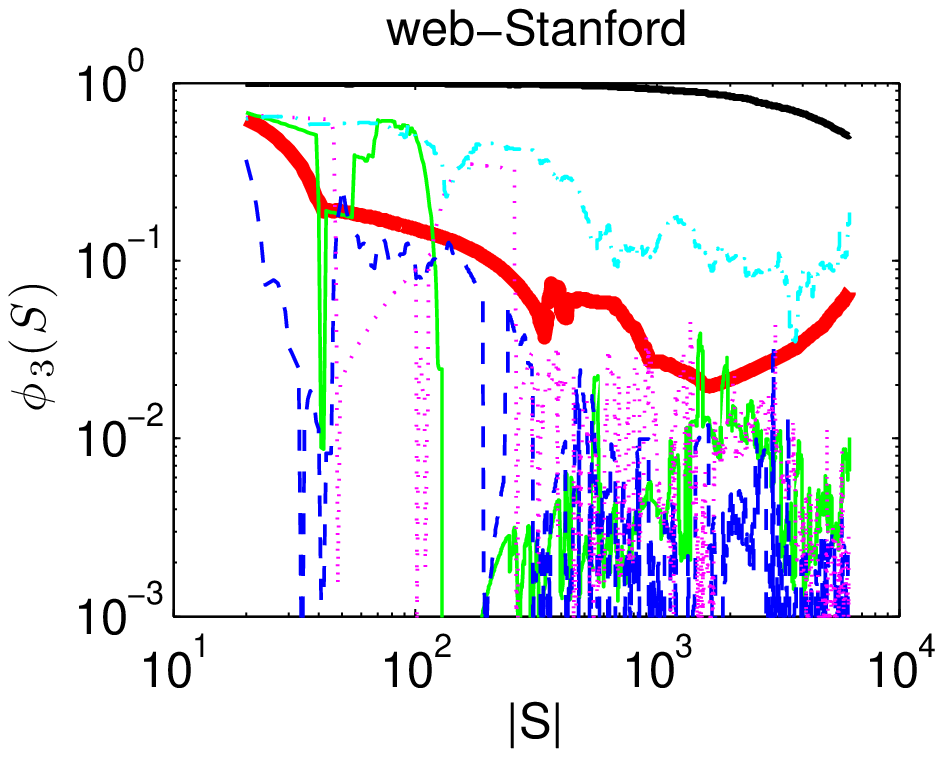} 
\includegraphics[height=3.4cm]{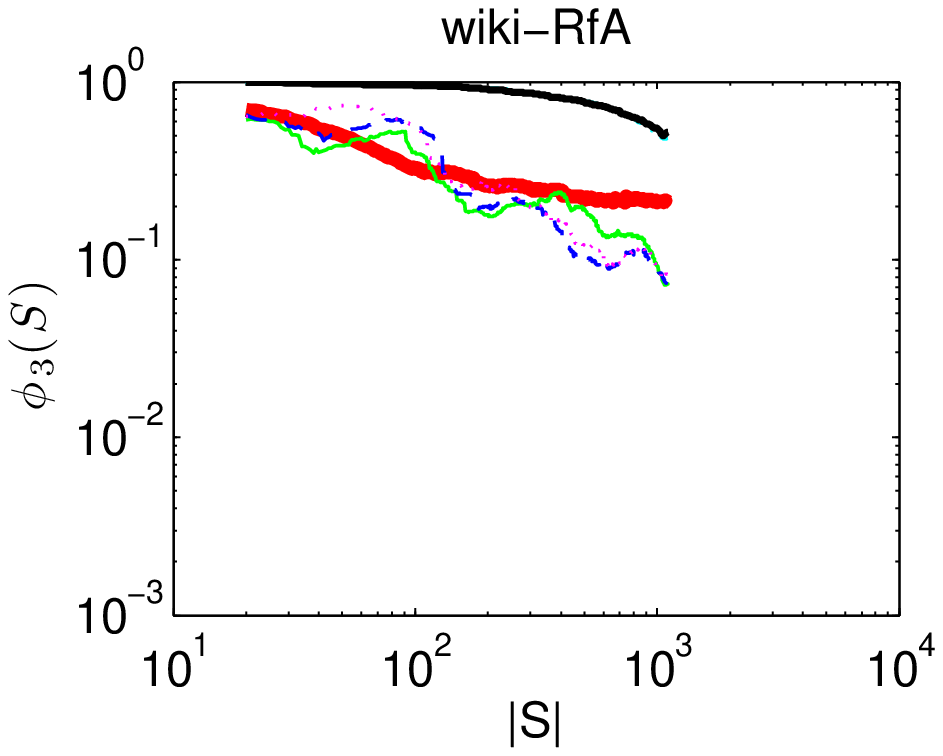}
\includegraphics[height=3.4cm]{wiki-Talk.eps}
\includegraphics[height=3.4cm]{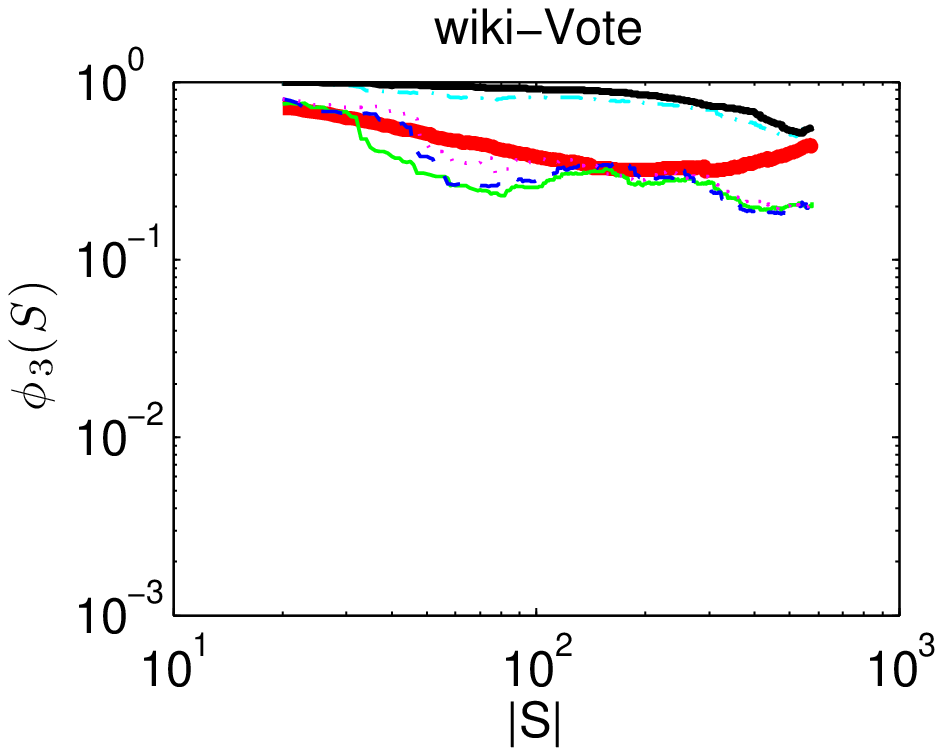}
\caption{
Directed 3-cycle higher-order conductance ($\phi_3(S)$, \eqnref{eqn:multi_cond}) as a function of the smaller partition size, $|S|$.
The size runs from twenty nodes to half the nodes in the network.
}
\label{fig:all_d3c_cond}
\end{figure*}

\begin{figure*}
\includegraphics[height=3.4cm]{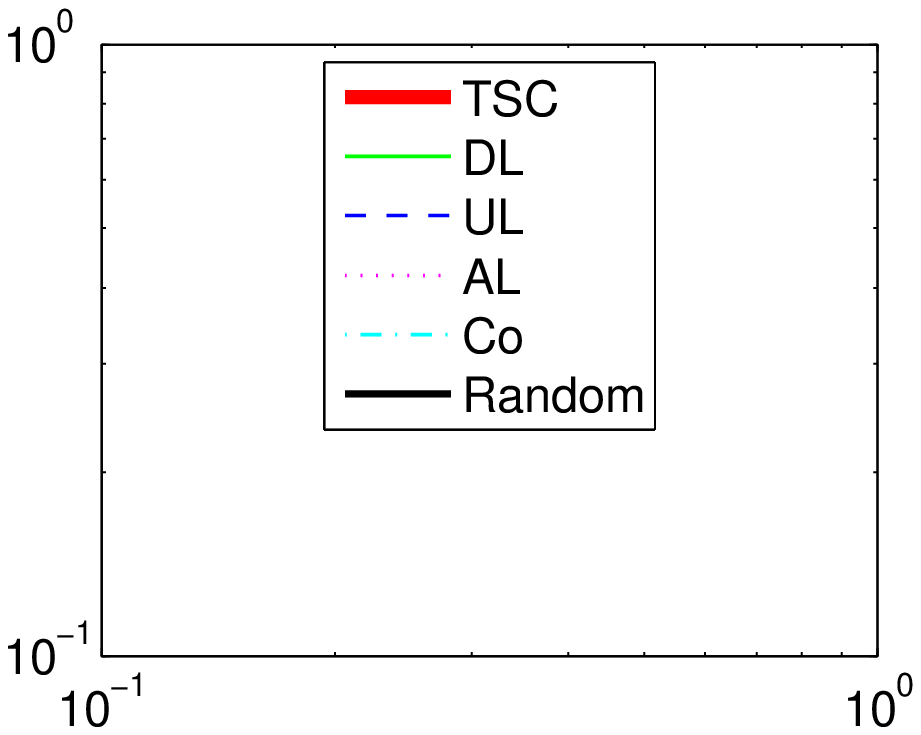}
\includegraphics[height=3.4cm]{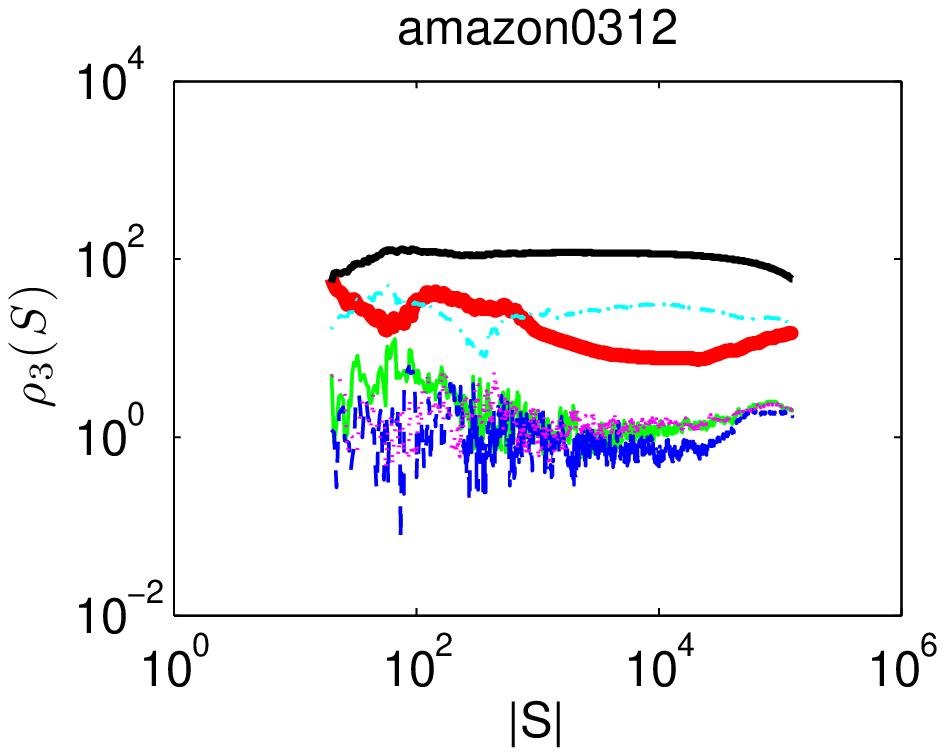}
\includegraphics[height=3.4cm]{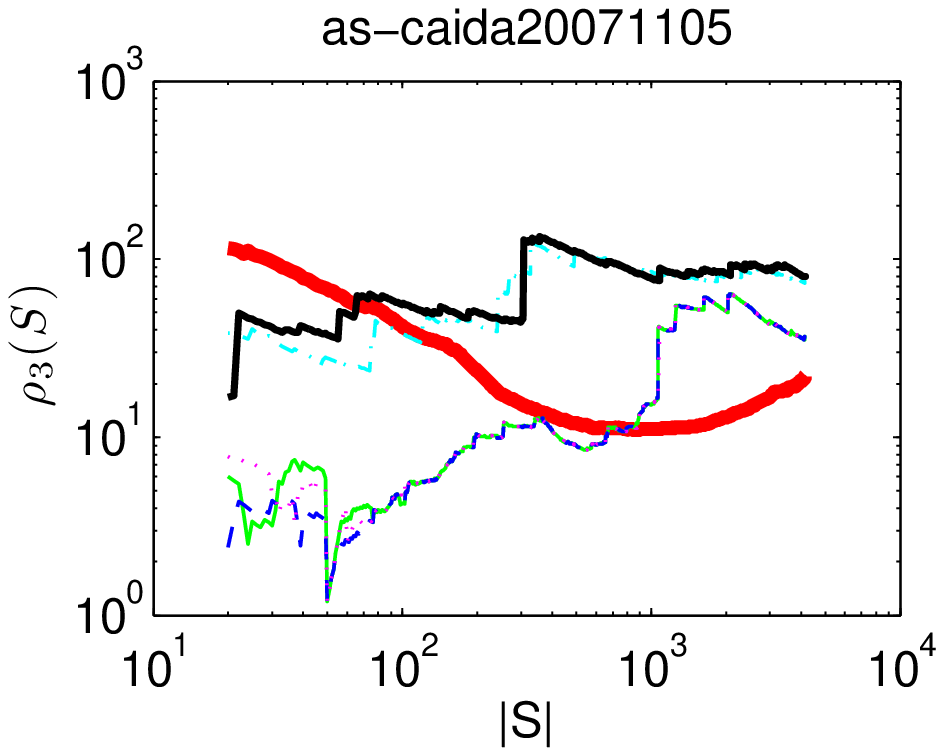}
\includegraphics[height=3.4cm]{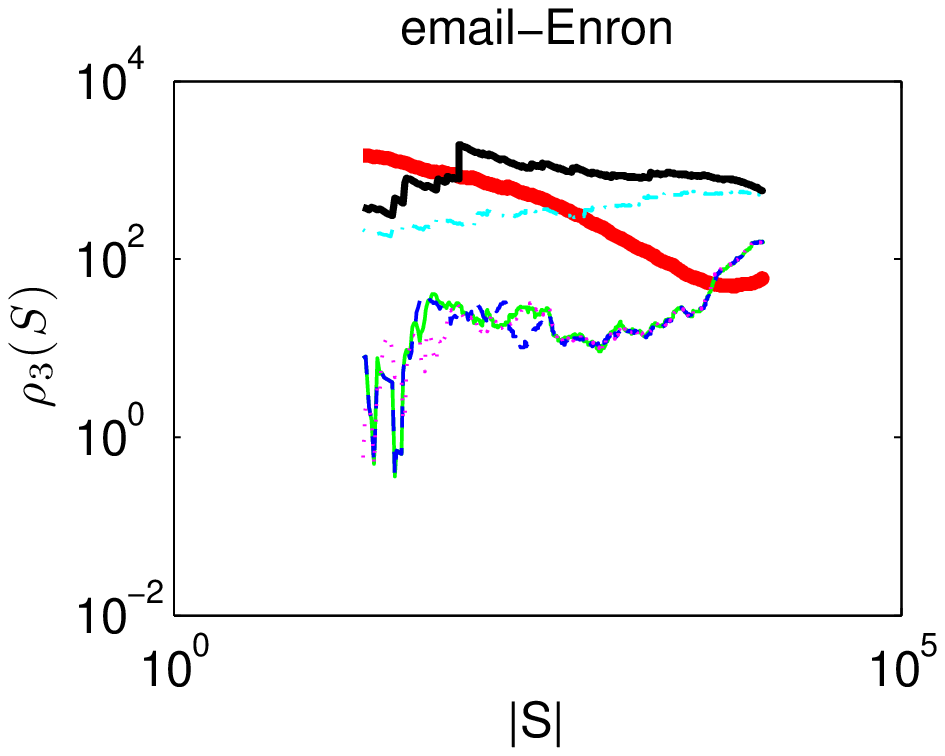} \\
\includegraphics[height=3.4cm]{email-EuAll_exp.eps}
\includegraphics[height=3.4cm]{soc-Epinions1_exp.eps}
\includegraphics[height=3.4cm]{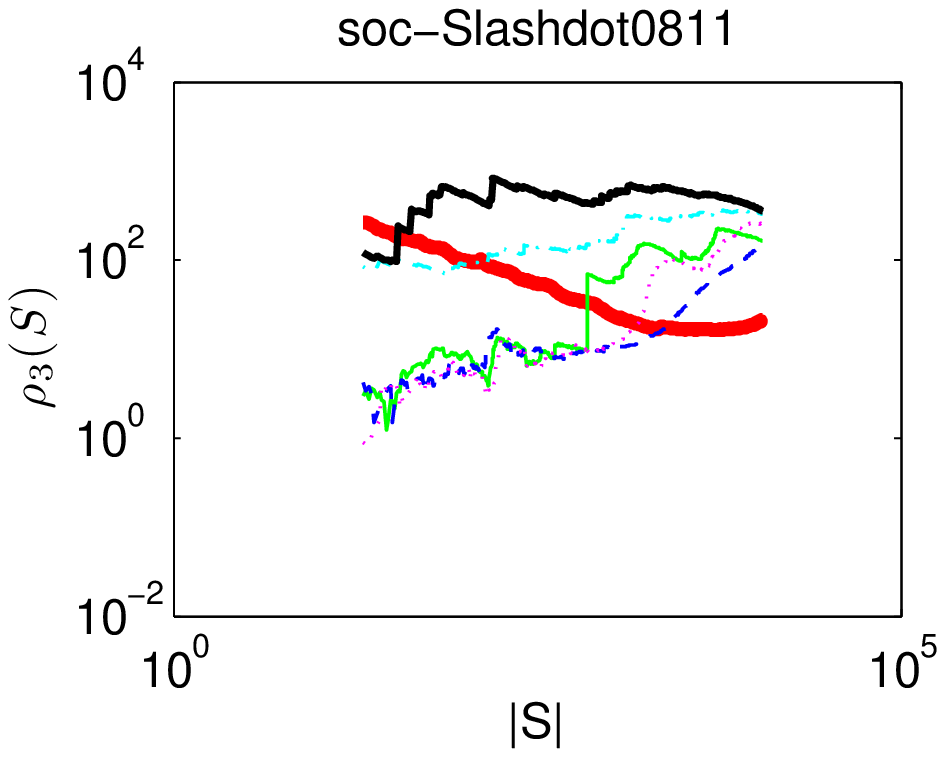}
\includegraphics[height=3.4cm]{twitter_combined_exp.eps} \\
\includegraphics[height=3.4cm]{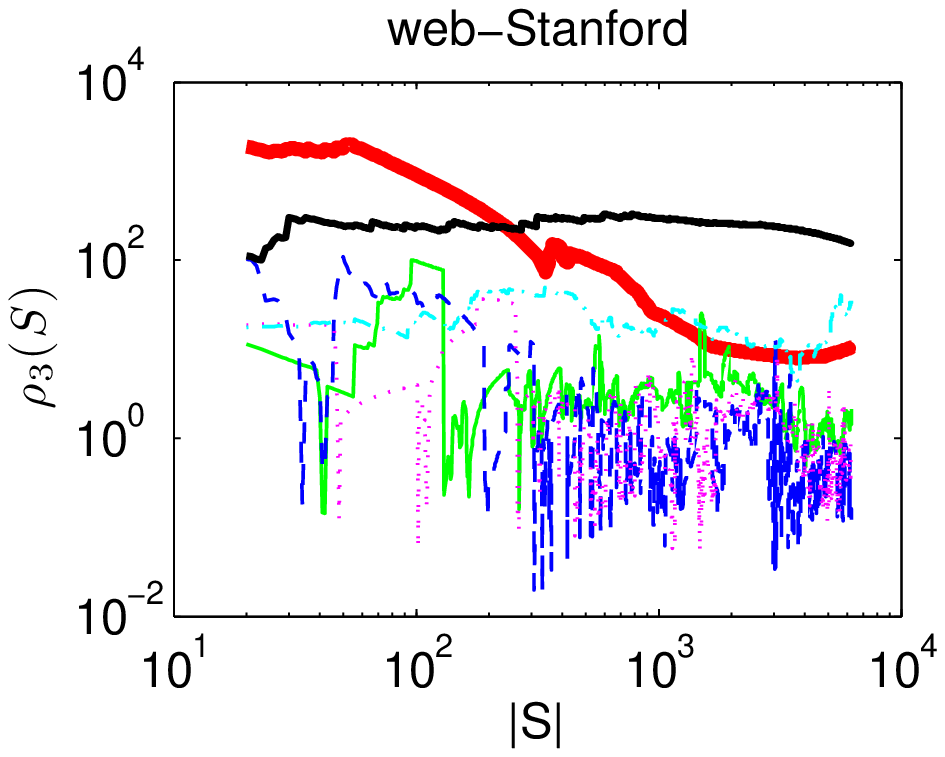} 
\includegraphics[height=3.4cm]{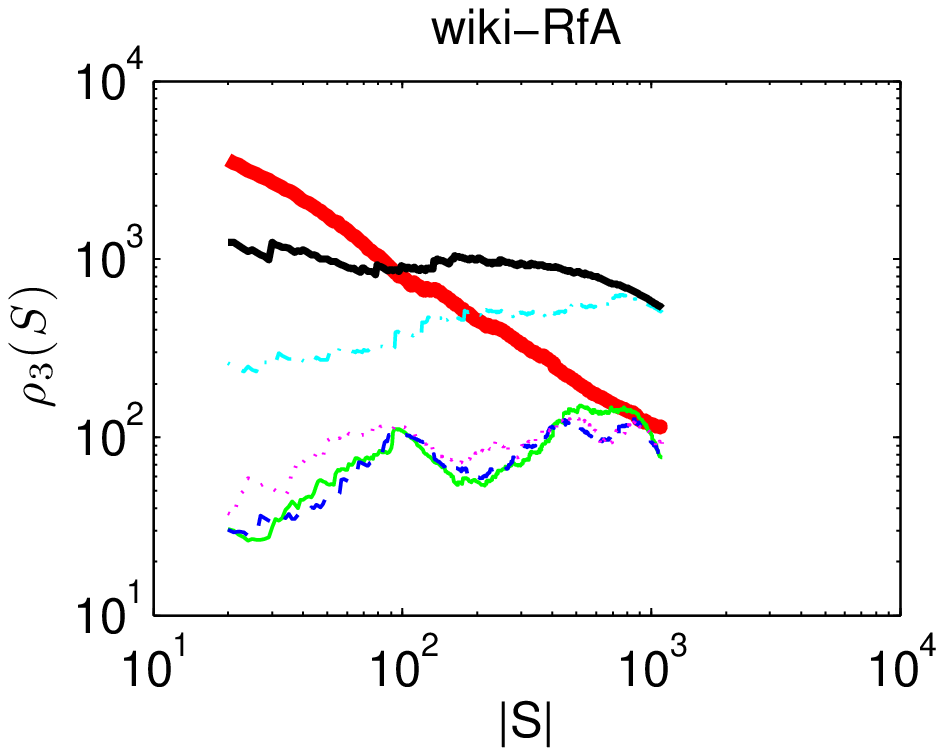}
\includegraphics[height=3.4cm]{wiki-Talk_exp.eps}
\includegraphics[height=3.4cm]{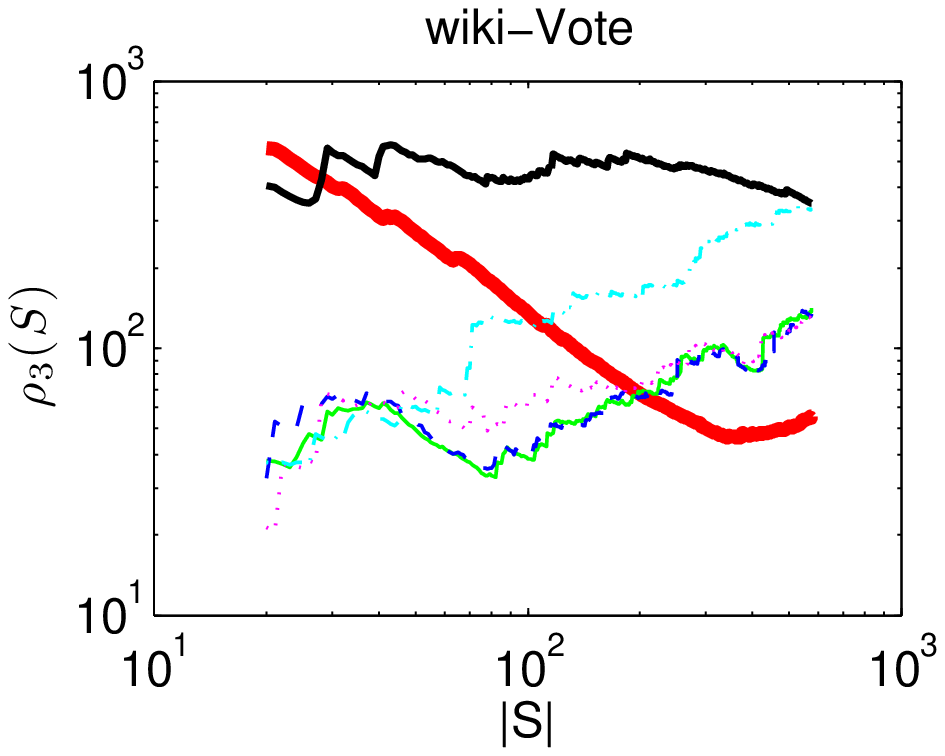}
\caption{
Directed 3-cycle higher-order expansion ($\rho_3(S)$, \eqnref{eqn:multi_cond}) as a function of the smaller partition size, $|S|$.
The size runs from twenty nodes to half the nodes in the network.
}
\label{fig:all_d3c_exp}
\end{figure*}

\begin{figure*}
\includegraphics[height=3.4cm]{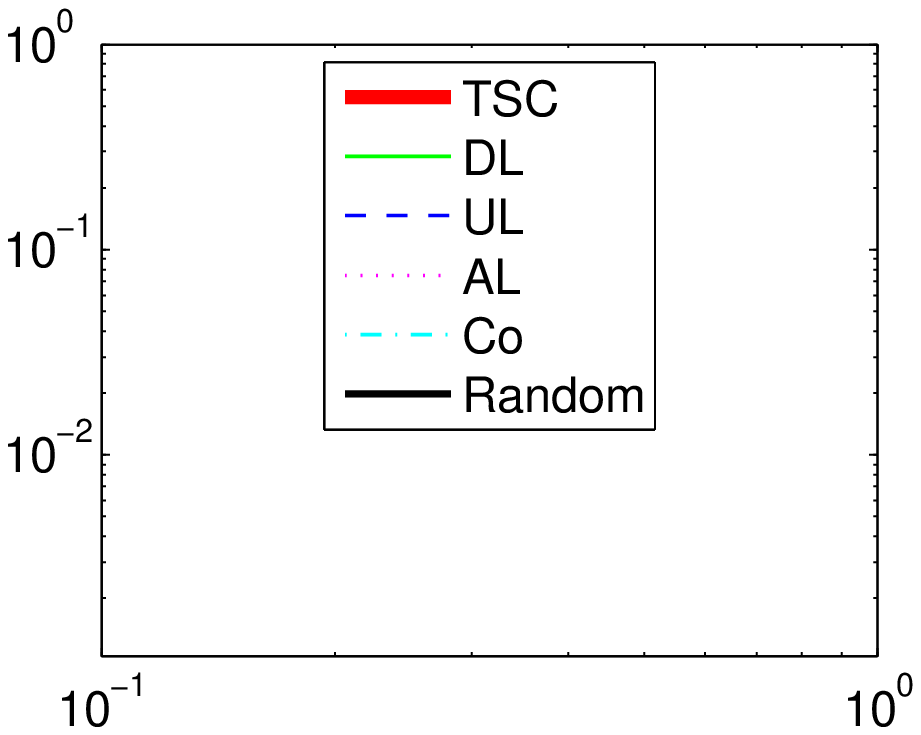}
\includegraphics[height=3.4cm]{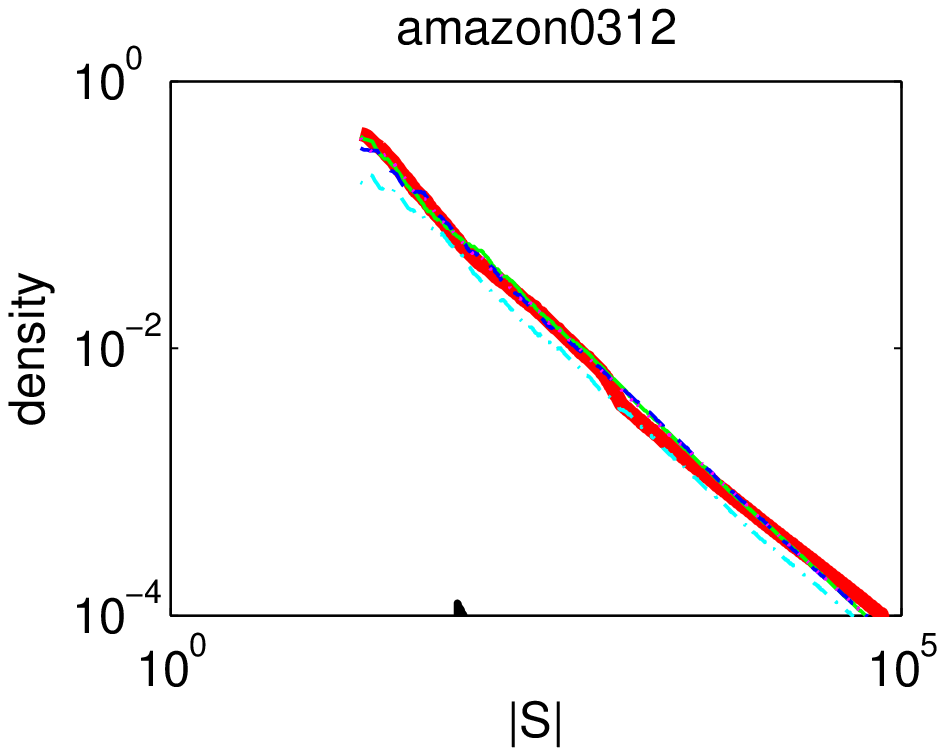}
\includegraphics[height=3.4cm]{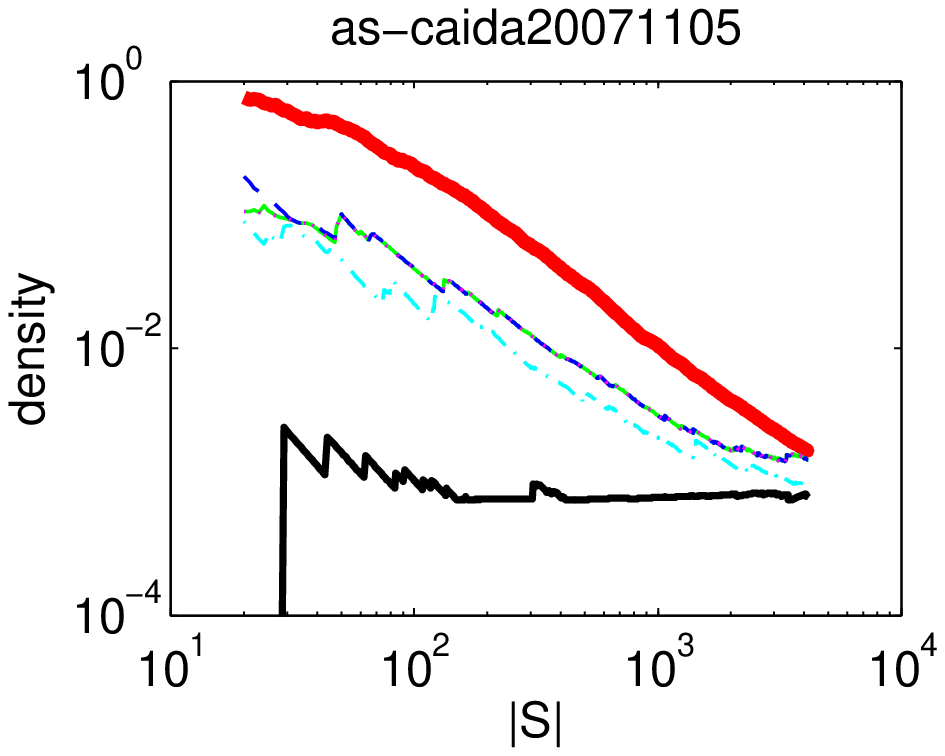}
\includegraphics[height=3.4cm]{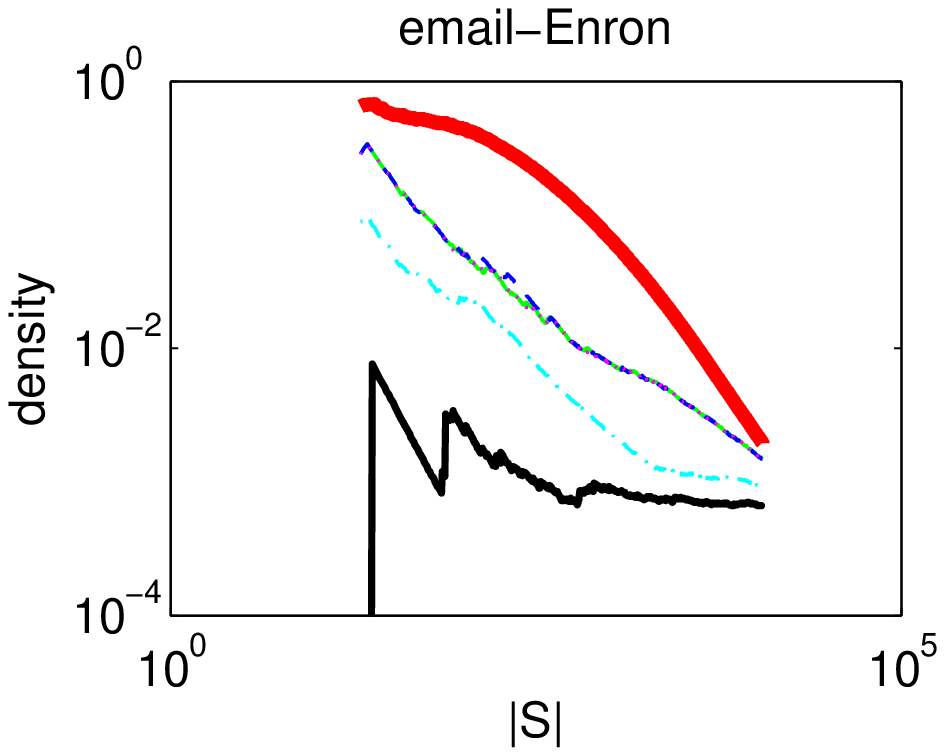} \\
\includegraphics[height=3.4cm]{email-EuAll_density.eps}
\includegraphics[height=3.4cm]{soc-Epinions1_density.eps}
\includegraphics[height=3.4cm]{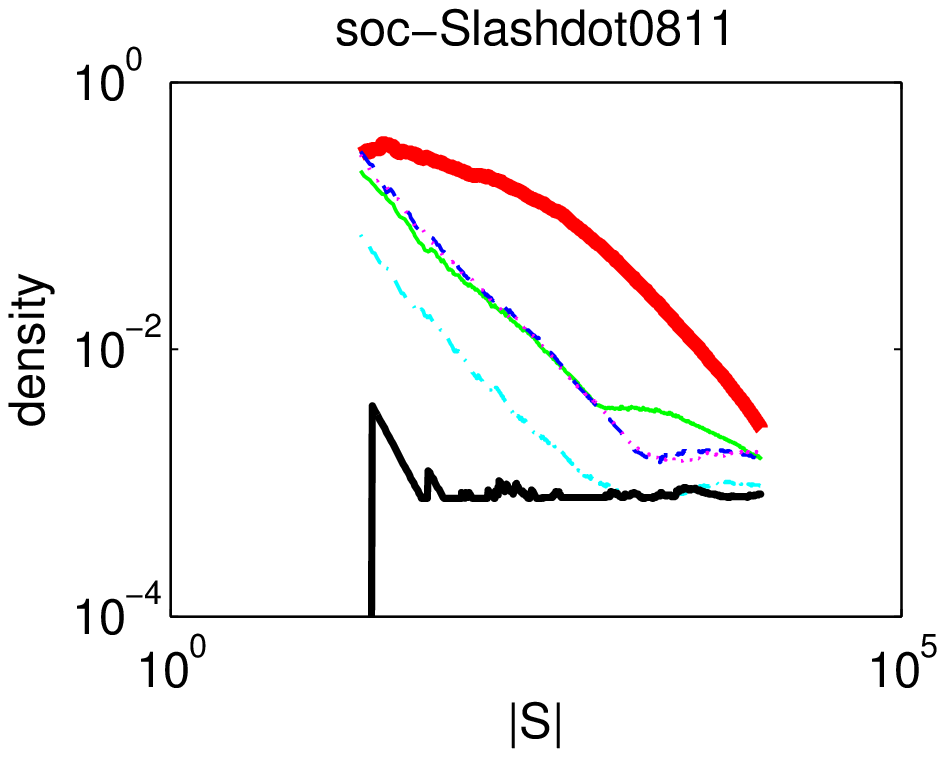}
\includegraphics[height=3.4cm]{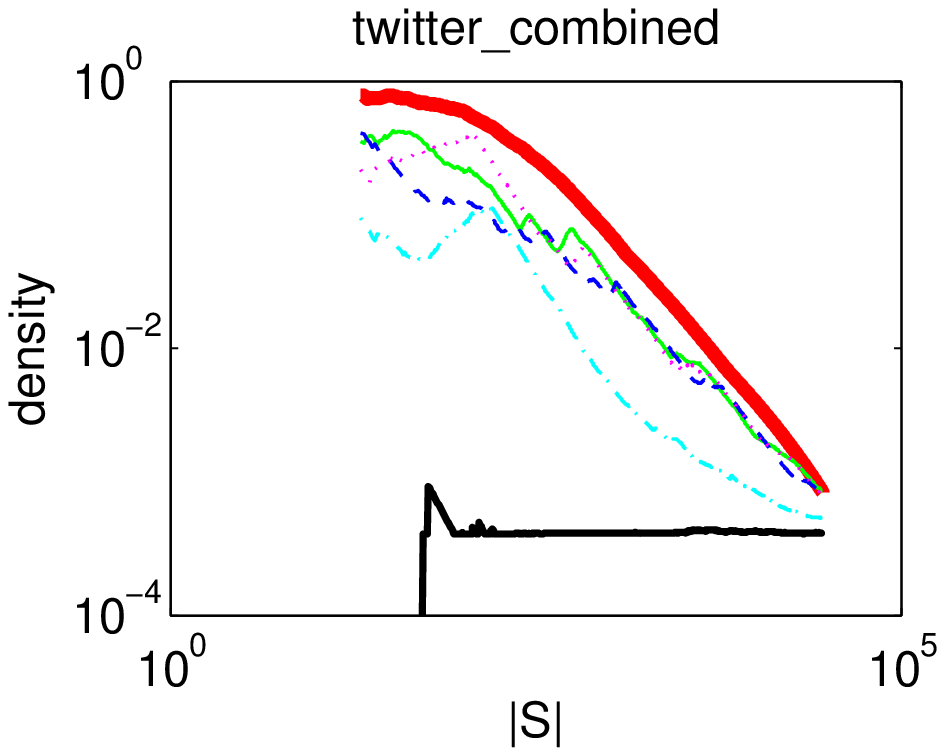} \\
\includegraphics[height=3.4cm]{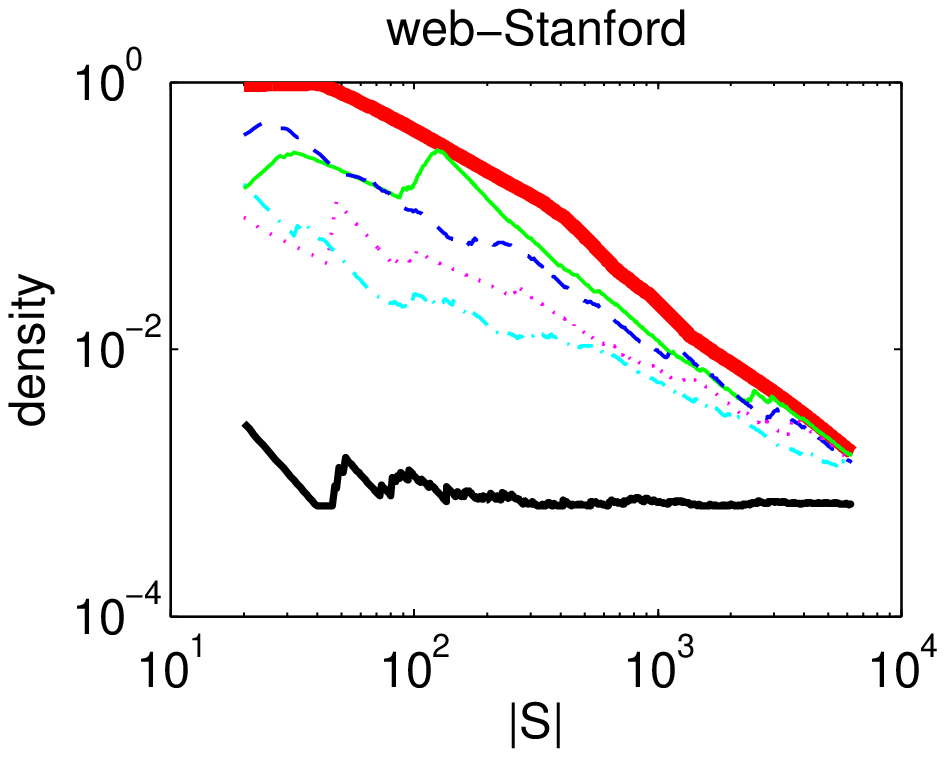} 
\includegraphics[height=3.4cm]{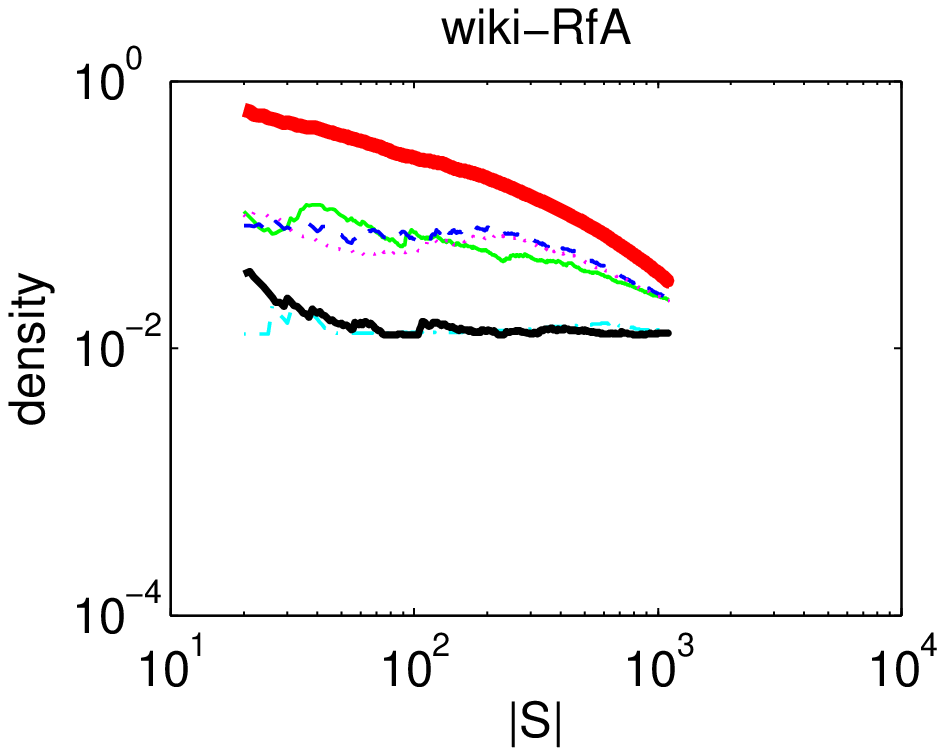}
\includegraphics[height=3.4cm]{wiki-Talk_density.eps}
\includegraphics[height=3.4cm]{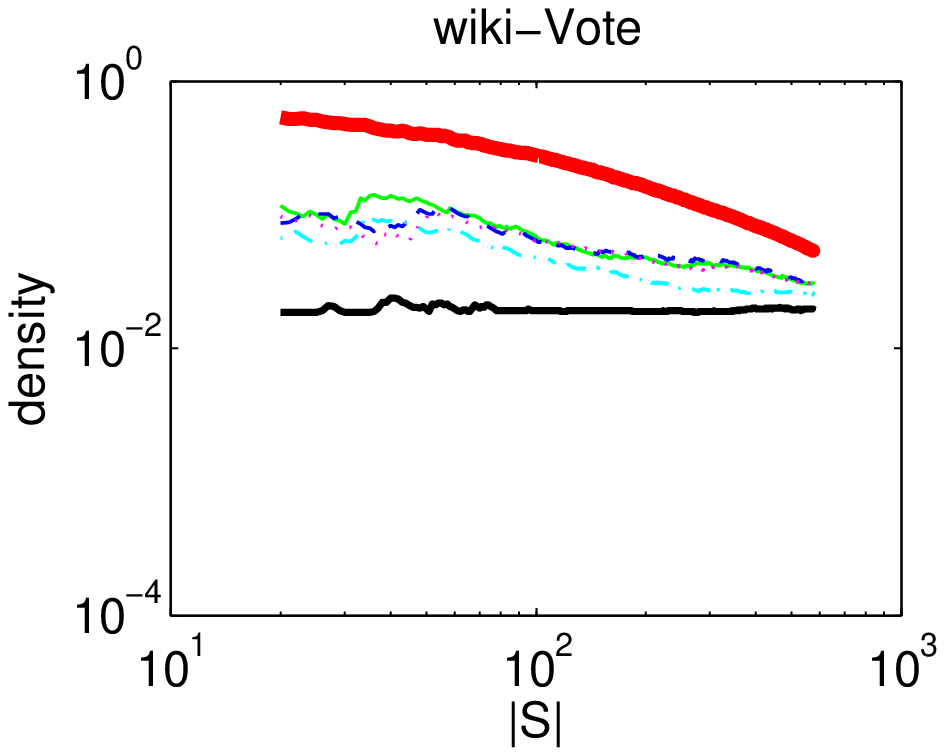}
\caption{
Density of the smaller partition set of vertices as a function of its size, $|S|$.
The size runs from twenty nodes to half the nodes in the network.
In nearly all cases, TSC finds the densest clusters.
}
\label{fig:all_density}
\end{figure*}

\end{document}